\documentclass[aps,prd,preprint,superscriptaddress,nofootinbib,
               showkeys,amsmath,amssymb]{revtex4-2}

\usepackage{graphicx}
\usepackage{mathtools}
\usepackage{amsthm}
\usepackage{xcolor}          
\usepackage[colorlinks=true,linkcolor=blue!60!black,
            citecolor=blue!60!black,urlcolor=blue!60!black]{hyperref}

\theoremstyle{plain}
\newtheorem{theorem}{Theorem}
\newtheorem{lemma}[theorem]{Lemma}
\newtheorem{proposition}[theorem]{Proposition}
\newtheorem{corollary}[theorem]{Corollary}
\theoremstyle{remark}
\newtheorem{remark}[theorem]{Remark}

\begin{document}

\title{Conformal flatness selects a universal isothermal attractor in
pure Lovelock gravity}

\author{Sudan Hansraj}
\email{hansrajs@ukzn.ac.za}
\affiliation{Astrophysics Research Centre,
School of Mathematics, Statistics and Computer Science,
University of KwaZulu-Natal, Private Bag X54001, Durban 4000,
South Africa}

\date{\today}

\begin{abstract}
We determine the complete solution set of the conformally flat
isotropic pure Lovelock field equations for every order $N\ge2$ and
every admissible dimension $d\ge2N+1$. Conformal flatness and pressure
isotropy combine into a single identity that factorises exactly,
splitting the solutions into two branches. The first is the
constant-density Schwarzschild interior, which persists at every order
with both metric potentials keeping their Einstein form. The second has
no Einstein counterpart and is governed by the single dimension--order
parameter $k=(d-2N)/[4(1-N)]$; we obtain its physically admissible
orbit in explicit closed parametric form, both potentials included. A
phase-space analysis on the projective line shows that no solution of
this branch has a pressure-free boundary at finite radius, so only the
Schwarzschild branch can describe a bounded star. The second instead
loses all memory of its central data and relaxes onto a pure Lovelock
isothermal sphere with $\rho\propto r^{-2N}$, a higher-curvature
analogue of the singular isothermal halo, which we show to be the
attractor of the whole family, approached at the closed-form rate
$\lambda_*=-(d-2)/(d-2N)$ and with limiting equation of state fixed by
$(d,N)$ alone. We prove that $p/\rho$ increases monotonically outward
there, and that of the two regular-centre orientations only one is
admissible. Since $k$ is rational the spatial potential is algebraic of
degree $u+v$, where $2k=-u/v$; radical inversion is guaranteed for
degree at most four, while four representative cases have full
symmetric Galois group.
\end{abstract}


\maketitle

\section{Introduction}
\label{sec:intro}

The search for exact solutions in higher-dimensional gravitational
theories has intensified alongside the development of string-theoretic
and braneworld models. The idea that extra spatial dimensions may play
a dynamical role goes back to the unification programme of
Kaluza~\cite{Kaluza1921} and Klein~\cite{Klein1926}, and and in recent times has been considered in braneworld cosmology~\cite{Maartens2010} and in the
systematic construction of second-order gravitational actions, of
which the scalar--tensor classification of Horndeski~\cite{Horndeski1974}
is the four-dimensional paradigm. In theories admitting extra
dimensions the Einstein--Hilbert action is naturally extended to
include higher-order curvature invariants, among which the Lovelock
polynomial action~\cite{Lovelock1971} occupies a distinguished
position. The Lovelock equations retain second-order field equations
in the metric, avoid ghosts, and reduce to Einstein gravity when the
higher-order coupling constants are switched off. This combination of
properties makes Lovelock gravity the most natural geometric extension
of general relativity to higher dimensions.

Within Lovelock gravity the so-called \emph{pure Lovelock}
framework~\cite{Dadhich2010} retains a single term of order $N$ in the
Lovelock polynomial, with all other coupling constants set to zero.
Pure Lovelock gravity exhibits a number of properties that parallel
those of Einstein gravity~\cite{Dadhich2012,Dadhich2016a}, including a
Schwarzschild-type vacuum solution and thermodynamic analogues, and it
is therefore an ideal framework in which to test how structural results of
standard general relativity generalise to higher dimension and higher
curvature.

The structural result at issue here is the classical uniqueness
theorem: the only static, spherically symmetric, conformally flat
perfect-fluid solution of the Einstein equations is the constant-density
Schwarzschild interior. The result was obtained by
Buchdahl~\cite{Buchdahl1971} from the vanishing of the Weyl tensor,
proved in the degenerate Petrov classes by Barnes~\cite{Barnes1972},
extended to the axistationary case by Collinson~\cite{Collinson1976},
and given a self-contained proof by Raychaudhuri and
Maiti~\cite{Raychaudhuri1979}; when expansion is permitted the
conformally flat perfect fluids are the Stephani
universes~\cite{Stephani1967}. Conformal flatness imposes the vanishing
of the Weyl tensor, a condition whose expression in terms of the metric
potentials is the same in every dimension~\cite{PoncedeLeon1987}.
Combined with the pure Lovelock pressure isotropy condition, it
generates a single master ordinary differential equation governing all
such spacetimes.

That the classical uniqueness fails once higher-curvature terms are
switched on was established for Einstein--Gauss--Bonnet gravity
in~\cite{Hansraj2021}, where the conformally flat static isotropic
system was shown to admit, besides the Schwarzschild interior, a second
metric that does not carry constant density. That analysis was
confined to Einstein--Gauss--Bonnet gravity, that is $N=2$ with the
Einstein term retained, and to the low dimensions $d=5,6$. The present
paper treats the pure Lovelock theory at arbitrary order $N$ and in
every admissible dimension; what is new here is the exact factorisation
of the combined system at arbitrary order, the resulting global
two-branch theorem, and an explicit closed parametric form of both
metric potentials on the physically admissible generic orbit. Earlier studies of exact
solutions in pure Lovelock
gravity~\cite{Maharaj2015,Hansraj2015a,Chilambwe2015,Hansraj2017,Molina2017}
have relied on ad hoc choices of one metric potential, or on
algorithmic generation of solutions, without addressing completeness at
all.

Both branches turn out to be of physical interest, and for different
reasons. The first describes bounded stars. The second does not: it
carries no pressure-free surface at any finite radius, and instead
approaches a fluid sphere with $\rho\propto r^{-2N}$ and $p/\rho$
constant. Configurations of exactly this kind are familiar in
four-dimensional Einstein gravity, where the singular isothermal
sphere $\rho\propto r^{-2}$ is the standard idealisation of an extended
self-gravitating halo and the profile responsible for asymptotically
flat rotation curves~\cite{BinneyTremaine,Chandrasekhar1939}. The
isothermal sphere retains a distinguished status in Lovelock gravity,
where its universality across orders was established
in~\cite{Dadhich2016a}.

The distinction between that earlier result and the present one should
be stated plainly. Reference~\cite{Dadhich2016a} shows that isothermal
spheres are universal \emph{solutions} of pure Lovelock gravity,
persisting at every order with $p=\alpha\rho$ and $\rho\sim r^{-2N}$ for
$d\ge2N+2$. What is established here is that, within the complete
conformally flat isotropic class, that same configuration is the
\emph{global large-radius attractor} of the physically admissible
branch. It is therefore not one specially chosen exact solution among
many: every regular-centre member of the family forgets its central
data at large radius and flows to it, at a convergence rate
$\lambda_*=-(d-2)/(d-2N)$ obtained here in closed form. Conformal
flatness, in other words, does not merely admit the isothermal sphere;
it selects it.

We show that the solution set consists of exactly two branches: the
constant-density Schwarzschild interior, which persists unchanged at
every Lovelock order and remains the only branch able to carry a finite
stellar boundary, and
a genuinely higher-order one-parameter family with no Einstein-gravity
counterpart, obtained here in closed implicit form. We identify
the families for which inversion by radicals is guaranteed and
establish Galois obstructions in four representative higher-degree
cases, and we supply a parametric representation that renders the
physical generic solution analytically tractable for every admissible
pair $(d,N)$.

The paper is organised as follows. Section~\ref{sec:field} records the
pure Lovelock field equations for a static spherically symmetric
perfect fluid. Section~\ref{sec:factorisation} imposes conformal
flatness, exhibits the exact factorisation of the combined system, and
derives the master equation of the generic branch together with an
explicit formula for the pressure. Section~\ref{sec:completeness}
identifies the singular branch with the Schwarzschild interior, proves
completeness of the two-branch structure, establishes that the generic
potential is always an algebraic function of the radial variable,
identifies all families for which radical solvability is guaranteed by
degree, and dentifies Galois obstructions in four representative higher-degree cases, with the proofs collected in Appendix A.
Section~\ref{sec:dynamical} develops the phase-space analysis on the
projective line, from which the absence of a finite pressure-free
boundary, the large-$r$ decay law and the isothermal attractor all
follow. Section~\ref{sec:parametric} constructs the parametric
representation adapted to the two fixed points and renders all
thermodynamic quantities analytic. Section~\ref{sec:physical} analyses
the physical conditions and presents numerical results, and
Section~\ref{sec:discussion} concludes.

\section{Pure Lovelock field equations}
\label{sec:field}

The Lovelock Lagrangian~\cite{Lovelock1971} is
\begin{equation}
  \mathcal{L} = \sum_{N=0}^{\lfloor(d-1)/2\rfloor}
                \alpha_N\,\mathcal{R}^{(N)},
\end{equation}
where
\begin{equation}
  \mathcal{R}^{(N)} = \frac{1}{2^N}
  \delta^{\mu_1\nu_1\cdots\mu_N\nu_N}_{\alpha_1\beta_1\cdots\alpha_N\beta_N}
  \prod_{s=1}^{N} R^{\alpha_s\beta_s}{}_{\mu_s\nu_s},
\end{equation}
and $\delta^{\mu_1\nu_1\cdots\mu_N\nu_N}_{\alpha_1\beta_1\cdots\alpha_N\beta_N}
=\frac{1}{N!}\delta^{\mu_1}_{[\alpha_1}\cdots\delta^{\nu_N}_{\beta_N]}$
is the generalised Kronecker delta. Variation with respect to the
metric gives the equations of motion~\cite{Dadhich2010}
\begin{equation}
  \sum_{n=0}^{N}\alpha_n\,G^{(n)}_{AB} = T_{AB},
\end{equation}
with $G^{(n)}_{AB}=n[R^{(n)}_{AB}-\tfrac12 R^{(n)}g_{AB}]$. The
framework contains the cosmological constant ($N=0$), Einstein gravity
($N=1$) and Gauss--Bonnet gravity ($N=2$) as special cases. Lovelock
analogues of the Riemann tensor and their algebraic properties are
developed in~\cite{Camanho2015,Dadhich2017b}, related geometric
constructions in Gauss--Bonnet gravity appear in~\cite{Maeda2007}, and
the sense in which gravitational dynamics exhibits universality across
Lovelock orders is reviewed in~\cite{Dadhich2013}. In the pure Lovelock
framework a single coupling is retained, so that
\begin{equation}
  G^{(N)}_{AB} = T_{AB}.
\label{eq:purelovelock}
\end{equation}
We set $\alpha_N=1$ throughout. Since $\alpha_N$ carries dimensions of
$(\text{length})^{2N-2}$, this normalisation fixes the unit of length;
Lemma~\ref{lem:scaling} below shows that this is the only role it
plays.

We consider the static $d$-dimensional spherically symmetric line
element
\begin{equation}
  ds^2 = e^{\nu}dt^2 - e^{\lambda}dr^2 - r^2 d\Omega^2_{d-2},
\label{eq:metric}
\end{equation}
with $\nu=\nu(r)$ and $\lambda=\lambda(r)$, sourced by a neutral
perfect fluid with comoving velocity $u^a = e^{-\nu/2}\delta^a_0$ and
$T^a{}_b=\mathrm{diag}(-\rho,p_r,p_\perp,\dots,p_\perp)$. The energy
density and radial pressure follow from~\eqref{eq:purelovelock}
as~\cite{Molina2017}
\begin{align}
  \rho &= \frac{(d-2)!\,e^{-\lambda}(1-e^{-\lambda})^{N-1}}
               {2(d-2N-1)!\,r^{2N}}
          \bigl[Nr\lambda' + (d-2N-1)(e^{\lambda}-1)\bigr],
\label{eq:density}\\
  p_r &= \frac{(d-2)!\,e^{-\lambda}(1-e^{-\lambda})^{N-1}}
              {2(d-2N-1)!\,r^{2N}}
         \bigl[Nr\nu' - (d-2N-1)(e^{\lambda}-1)\bigr],
\label{eq:pressure}
\end{align}
primes denoting $d/dr$. Setting $N=1$ and $d=4$
returns the standard Einstein expressions. The factorials
in~\eqref{eq:density} and~\eqref{eq:pressure} require $d\ge2N+1$, which
is also the range in which the $N$th Lovelock term is dynamically
non-trivial; throughout what follows we therefore work on
\begin{equation}
  N\ge2,\qquad d\ge 2N+1,
\label{eq:domain}
\end{equation}
treating $N=1$ separately as the Einstein limit. The critical odd
dimension $d=2N+1$ is retained throughout, its distinctive behaviour
being derived where it arises; that pure Lovelock gravity admits no
bound distribution of finite radius in this dimension is already known
from~\cite{Dadhich2017}, and the statements below concerning $d=2N+1$
are consistent with, rather than independent of, that result. The conservation law
$T^{ab}{}_{;b}=0$ reads
\begin{equation}
  \tfrac12(p_r+\rho)\nu' + p_r' + \frac{d-2}{r}\left(p_r-p_\perp \right)=0,
\label{eq:conservation}
\end{equation}
so that the pressure isotropy condition $p_r=p_\perp\equiv p$ is
equivalent to $\tfrac12(p+\rho)\nu' + p'=0$ with $p$ given
by~\eqref{eq:pressure}.

It is advantageous to make use of the Buchdahl variables
\begin{equation}
  x = Cr^2,\qquad e^{-\lambda}=Z(x),\qquad e^{\nu}=y^2(x),
\label{eq:transform}
\end{equation}
with $C>0$ throughout, so that $x$ increases with the areal radius and
ranges over $(0,\infty)$; this sign choice is assumed wherever
$\tau=\ln x$ is used below and in statements such as the positivity
of~\eqref{eq:rhoconst} for $\beta<0$. In these variables
\eqref{eq:density}--\eqref{eq:pressure} and the isotropy condition
become
\begin{align}
  \rho &= \frac{C^N(d-2)!(1-Z)^{N-1}}{2(d-2N-1)!\,x^N}
          \bigl[(d-2N-1)(1-Z)-2Nx\dot{Z}\bigr],
\label{eq:rho_Z}\\
  p &= \frac{C^N(d-2)!(1-Z)^{N-1}}{2(d-2N-1)!\,x^N y}
       \bigl[4NxZ\dot{y}-(d-2N-1)(1-Z)y\bigr],
\label{eq:p_Z}
\end{align}
\begin{align}
  0 &= 4x^2Z(1-Z)\ddot{y} \nonumber\\
    &\quad + \bigl[4(1-N)xZ(1-Z)+2x^2(1-(2N-1)Z)\dot{Z}\bigr]\dot{y}
      \nonumber\\
    &\quad + (d-2N-1)(1-Z)\bigl[x\dot{Z}-Z+1\bigr]y,
\label{eq:isotropy_Z}
\end{align}
where dots denote $d/dx$. The decisive advantage
of~\eqref{eq:transform} is that~\eqref{eq:isotropy_Z} is \emph{linear}
in $y$. Setting $N=1$ recovers the $d$-dimensional Einstein isotropy
equation
\begin{equation}
  4x^2Z\ddot{y}+2x^2\dot{Z}\dot{y}+(d-3)(x\dot{Z}-Z+1)y=0.
\label{eq:einstein_iso}
\end{equation}

\section{Conformal flatness and the exact factorisation}
\label{sec:factorisation}

\subsection{The conformal flatness condition}

The vanishing of the Weyl tensor for the metric~\eqref{eq:metric} is
governed by a condition that, remarkably, is independent of the
spacetime dimension~\cite{PoncedeLeon1987}: the potentials must satisfy
\begin{equation}
  r^2(2\nu''+\nu'^2-\nu'\lambda')-r(\nu'-\lambda')-2(e^{\lambda}-1)=0,
\label{eq:conformal_r}
\end{equation}
with integrated form
\begin{equation}
  e^{\nu}=C_1r^2\cosh^2\!\Bigl(\int\frac{e^{\lambda/2}}{r}\,dr+C_2\Bigr).
\end{equation}
Under~\eqref{eq:transform}, equation~\eqref{eq:conformal_r} becomes
\begin{equation}
  4x^2Z\ddot{y}+2x^2\dot{Z}\dot{y}-(x\dot{Z}-Z+1)y=0,
\label{eq:conformal_Z}
\end{equation}
with general integral
\begin{equation}
  y = A\sqrt{x}\,\cosh\Bigl(\tfrac12\!\int\frac{dx}{x\sqrt{Z}}+B\Bigr).
\label{eq:y_integral}
\end{equation}
Writing $A\equiv x\dot{Z}-Z+1$, the two second-order
equations~\eqref{eq:einstein_iso} and~\eqref{eq:conformal_Z} share
their derivative terms and differ only in the undifferentiated one,
carrying $+(d-3)Ay$ and $-Ay$ respectively. Subtracting them gives
\begin{equation}
  (d-2)\,A\,y=0,
\label{eq:einstein_compat}
\end{equation}
so that $A=0$ for every $d>2$. This is the higher-dimensional
uniqueness theorem in one line, and is recovered below as
Corollary~\ref{cor:einstein}.

\subsection{The factorisation}

Multiplying~\eqref{eq:conformal_Z} by $(1-Z)$ and subtracting the
result from~\eqref{eq:isotropy_Z}, the second-derivative terms cancel
identically and what remains factorises completely,
\begin{equation}
  \bigl(x\dot{Z}-Z+1\bigr)
  \bigl[4(1-N)xZ\dot{y}+(d-2N)(1-Z)y\bigr]=0.
\label{eq:factorised}
\end{equation}
The solution set therefore splits into exactly two branches. The first
factor vanishes on the locus $x\dot{Z}=Z-1$, whose general solution is
$Z=1+\beta x$; we call this the \emph{singular branch} and identify it
with the constant-density Schwarzschild interior in
Theorem~\ref{thm:schwbranch}. On its complement --- the \emph{generic
branch} --- the second factor must vanish, giving
\begin{equation}
  \frac{\dot{y}}{y}=\frac{(d-2N)(Z-1)}{4(1-N)xZ}.
\label{eq:ydot_over_y}
\end{equation}
This is the key reduction: on the generic branch the logarithmic
derivative of the temporal potential is expressed in terms of $Z$
alone. At $N=1$ the second factor of~\eqref{eq:factorised} degenerates
to $(d-2)(1-Z)y=0$, so the generic branch is empty and the singular
branch exhausts the solution set; for $N\ge2$ both branches are
populated.

Substituting~\eqref{eq:ydot_over_y} into~\eqref{eq:p_Z} gives the
explicit pressure formula
\begin{equation}
  p = \frac{C^N(d-2)!\,(d-N-1)\,(1-Z)^N}{2(d-2N-1)!\,(N-1)\,x^N}.
\label{eq:pressure_explicit}
\end{equation}
This is a central structural result: \emph{on the generic branch the
pressure of every conformally flat isotropic pure Lovelock spacetime is
determined entirely by the spatial potential $Z$, independently of the
temporal potential $y$.} The decoupling parallels a known feature of
the Finch--Skea class~\cite{Finch1989,Hansraj2015b} in Einstein gravity.

\subsection{The master equation}

Differentiating~\eqref{eq:ydot_over_y} and eliminating $\dot y/y$ gives
\begin{equation}
  \frac{\ddot{y}}{y}
  = \frac{kx\dot{Z}-kZ(Z-1)+k^2(Z-1)^2}{x^2Z^2},
\label{eq:yddot_over_y}
\end{equation}
in which the entire dependence on $d$ and $N$ has collapsed into the
single \emph{dimension--order parameter}
\begin{equation}
  k \equiv \frac{d-2N}{4(1-N)}.
\label{eq:k_def}
\end{equation}
Substituting~\eqref{eq:ydot_over_y} and~\eqref{eq:yddot_over_y}
into~\eqref{eq:conformal_Z} and clearing a factor $Z$ yields the
\emph{master equation}
\begin{equation}
  \bigl[2k+(2k-1)Z\bigr]x\dot{Z}
  = (Z-1)\bigl[4k^2-(2k-1)^2Z\bigr].
\label{eq:master}
\end{equation}
This single first-order equation governs the spatial potential of every
generic-branch conformally flat isotropic pure Lovelock metric, for
every admissible $N\ge2$, $d\ge2N+1$, entirely through $k$. In canonical form it is an Abel
equation of the second kind; because its coefficients are constant it
is autonomous in $\tau=\ln x$ and hence separable, so no
Abel-specific machinery is required and the integration below succeeds
uniformly in $k$.

Throughout what follows we write
\begin{equation}
  L(Z)\equiv(2k-1)^2Z-4k^2 ,
\label{eq:Ldef}
\end{equation}
so that $L(Z^*)=0$ and, since $dL/dZ=(2k-1)^2>0$, $L>0$ strictly on the
physical arc $Z^*<Z<1$, where $Z-1<0$.

\begin{theorem}[General non-equilibrium solution]
\label{thm:general}
The general non-constant solution of~\eqref{eq:master} satisfies
\begin{equation}
  \frac{|Z-1|}{x}
  \left[\frac{x\,|L(Z)|}{|Z-1|}\right]^{2k}=c_1 ,
  \qquad c_1>0 ,
\label{eq:implicit_abs}
\end{equation}
and in particular, on the physical arc $Z^*<Z<1$,
\begin{equation}
  \frac{1-Z}{x}\left[\frac{x\,L(Z)}{1-Z}\right]^{2k}=c_1 ,
\label{eq:implicit}
\end{equation}
where $c_1$ is the constant of integration.
\end{theorem}

\begin{proof}
Separating variables in~\eqref{eq:master},
\begin{equation}
  \frac{2k+(2k-1)Z}{(Z-1)\bigl[4k^2-(2k-1)^2Z\bigr]}\,dZ
  = \frac{dx}{x},
\label{eq:separated}
\end{equation}
and decomposing the left-hand side into partial fractions gives
\begin{equation}
  \frac{1}{Z-1}+\frac{2k(2k-1)}{4k^2-(2k-1)^2Z},
\end{equation}
which integrates to
\begin{equation}
  \ln|Z-1|-\frac{2k}{2k-1}\ln|L(Z)|
  = \ln|x|+\text{const}.
\end{equation}
Exponentiating and rearranging yields~\eqref{eq:implicit_abs}, the
absolute values being those generated by the logarithms. On the
physical arc $|Z-1|=1-Z$ and $|L|=L$, which gives~\eqref{eq:implicit}.
\qed
\end{proof}

The absolute values are important. On the physical arc
$x\,L(Z)/(Z-1)<0$, and $2k$ is non-integral for most admissible pairs
--- for instance $2k=-\tfrac12$ at $(N,d)=(2,5)$ and $2k=-\tfrac14$ at
$(3,7)$ --- so the unsigned combination would require a fractional power
of a negative quantity. Written as~\eqref{eq:implicit} the integral is
real and positive throughout the physical range, and is the form
consistent with the parametrisation of Sect.~\ref{sec:parametric},
where $s>0$.

The separation~\eqref{eq:separated} divides by the two factors of the
right-hand side of~\eqref{eq:master}, so~\eqref{eq:implicit_abs}
describes the non-equilibrium trajectories only. The equation also admits the two
constant solutions $Z\equiv1$ and $Z\equiv Z^*=4k^2/(2k-1)^2$, which are
the fixed points of the flow studied in Sect.~\ref{sec:dynamical} and
are not contained in~\eqref{eq:implicit_abs} for any finite non-zero
$c_1$. They are recorded explicitly in Theorem~\ref{thm:nogo}: the
first coincides with the $\beta=0$ member of the singular branch, and
the second, integrating~\eqref{eq:ydot_over_y} at constant $Z$, is the
exact power-law solution
\begin{equation}
  Z\equiv Z^*,\qquad y=y_0\,x^{(4k-1)/(4k)},
\label{eq:fixedsol}
\end{equation}
analysed in Sect.~\ref{ssec:isothermal}.

The constant $c_1$ labels the members of the family, but it carries no
invariant content, as the following observation shows.

\begin{lemma}[Scaling orbit]
\label{lem:scaling}
The regular non-constant family in $Z^*<Z<1$ is a single orbit of the
radial scaling group $x\mapsto\mu x$, $\mu>0$. Under this action $Z$ is
unchanged in value, while
\begin{equation}
  c_1\mapsto\mu^{\,2k-1}c_1,\qquad
  (\rho,p)\mapsto\mu^{-N}(\rho,p).
\label{eq:scalingweights}
\end{equation}
Since $2k-1=(d-2)/[2(1-N)]\neq0$ for every $d>2$, distinct rescalings produce distinct values of \(c_1\), and the scaling group acts simply transitively on the family. The
equilibria of Theorem~\ref{thm:general} are fixed points of the action
and are not members of the orbit.
\end{lemma}

\begin{proof}
Equation~\eqref{eq:master} involves $x$ only through the combination
$x\dot{Z}=dZ/d\tau$, $\tau=\ln x$, and is therefore invariant under
$x\mapsto\mu x$; if $Z(x)$ is a solution then so is $Z(x/\mu)$.
Substituting $x\mapsto\mu x$ at fixed $Z$ into~\eqref{eq:implicit_abs}
multiplies the left-hand side by $\mu^{2k-1}$, and the same
substitution in~\eqref{eq:rho_Z} and~\eqref{eq:pressure_explicit},
whose only explicit $x$ dependence is the factor $x^{-N}$, multiplies
$\rho$ and $p$ by $\mu^{-N}$. Since \(2k-1\neq0\), no non-trivial rescaling leaves \(c_1\) unchanged. \qed
\end{proof}

The regular physical family therefore possesses a universal
dimensionless profile, the integration constant fixing only the radial
and density scales. In the notation of Sect.~\ref{ssec:central}, where
$1-Z=ax+\mathcal{O}(x^2)$ near the centre, the transformed solution
$Z(x/\mu)$ has $1-Z(x/\mu)=(a/\mu)x+\mathcal{O}(x^2)$, so the group
acts by
\begin{equation}
  a\longmapsto a/\mu ,
\label{eq:aaction}
\end{equation}
consistently with $\rho_c\propto a^N$ and the weight $\mu^{-N}$
in~\eqref{eq:scalingweights}. The solutions satisfy $Z_a(x)=Z_1(ax)$,
so a single numerical integration determines the family for each
$(d,N)$ up to this rescaling. The transformation is a homothety rather than an isometry:
the central density~\eqref{eq:centralvals} does change under it once a
physical length normalisation has been fixed, and the dimensionless
ratios of Sect.~\ref{ssec:central} are what the lemma renders
universal.

\section{The two branches, completeness and solvability}
\label{sec:completeness}

\subsection{The Schwarzschild interior branch}

\begin{corollary}[Einstein limit]
\label{cor:einstein}
For $N=1$ the parameter $k$ is undefined. Solving
\eqref{eq:isotropy_Z} and~\eqref{eq:conformal_Z} simultaneously at
$N=1$ gives $Z=1+\beta x$ with $y=a\sqrt{1+\beta x}+b$ for
$\beta\neq0$, the Schwarzschild interior solution, together with the
degenerate case $\beta=0$, $y=ax+b$. The latter is not in general a
vacuum: at $N=1$ the factor $(1-Z)^{N-1}$ of~\eqref{eq:rho_Z}
and~\eqref{eq:p_Z} is unity rather than zero, and $Z\equiv1$ gives
\begin{equation}
  \rho=0,\qquad p=\frac{2C(d-2)\,a}{ax+b},
\label{eq:N1degenerate}
\end{equation}
so this configuration carries pressure without density unless $a=0$;
only the constant-lapse member $a=0$ is a genuine vacuum. For $N\ge2$
the factor $(1-Z)^{N-1}$ vanishes at $Z\equiv1$ and both $\rho$ and $p$
are zero, which is the case recorded in
Theorem~\ref{thm:schwbranch}. Up to this degenerate zero-density
solution, the Schwarzschild interior is the unique conformally flat
isotropic solution in Einstein gravity in every dimension $d$.
\end{corollary}

\begin{theorem}[The singular branch]
\label{thm:schwbranch}
For every admissible pair $N\ge2$, $d\ge2N+1$ of~\eqref{eq:domain},
the singular branch consists of $Z=1+\beta x$ together with
\begin{equation}
  y=\begin{cases}
      a\sqrt{1+\beta x}+b, & \beta\neq0,\\[1mm]
      a\,x+b,               & \beta=0,
    \end{cases}
\label{eq:schwbranch}
\end{equation}
both cases satisfying the isotropy condition~\eqref{eq:isotropy_Z} and
the conformal flatness condition~\eqref{eq:conformal_Z}. For
$\beta\neq0$ this is the constant-density Schwarzschild interior, with
\begin{equation}
  \rho=\frac{(d-1)!\,(-\beta C)^N}{2(d-2N-1)!},
\label{eq:rhoconst}
\end{equation}
positive for $\beta<0$. For $\beta=0$ the spatial sections are flat,
$Z\equiv1$, and since $N\ge2$ the factor $(1-Z)^{N-1}$
in~\eqref{eq:rho_Z} and~\eqref{eq:p_Z} gives $\rho=p=0$: a degenerate
conformally flat pure Lovelock vacuum. Outside the stated domain, at
$N=1$, that factor is unity instead and the same potentials give
$\rho=0$ with $p\neq0$; see Corollary~\ref{cor:einstein}.
\end{theorem}

\begin{proof}
On the ansatz $Z=1+\beta x$ the factor $x\dot{Z}-Z+1$ vanishes
identically, so the undifferentiated $y$-terms of
both~\eqref{eq:isotropy_Z} and~\eqref{eq:conformal_Z} drop out, and
\eqref{eq:conformal_Z} reduces to $2(1+\beta x)\ddot{y}+\beta\dot{y}=0$.
In~\eqref{eq:isotropy_Z} the $\ddot{y}$ coefficient becomes
$-4\beta x^3(1+\beta x)$ while the $\dot{y}$ coefficient is
\begin{equation}
  4(1-N)x(1+\beta x)(-\beta x)
  +2x^2\beta\bigl[(2-2N)-(2N-1)\beta x\bigr]=-2\beta^2x^3,
\end{equation}
the constant and $\mathcal{O}(\beta x)$ parts cancelling identically, so
that~\eqref{eq:isotropy_Z} becomes
$-2\beta x^3\bigl[2(1+\beta x)\ddot{y}+\beta\dot{y}\bigr]=0$.

For $\beta\neq0$ both conditions therefore collapse to the same equation
$2(1+\beta x)\ddot{y}+\beta\dot{y}=0$, independently of $d$ and $N$,
whose general integral is the first line of~\eqref{eq:schwbranch};
substitution into~\eqref{eq:rho_Z} gives~\eqref{eq:rhoconst}.

For $\beta=0$ the two conditions decouple. Every term
of~\eqref{eq:isotropy_Z} carries a factor $(1-Z)$ or $\dot{Z}$, both of
which vanish identically at $Z\equiv1$, so isotropy is satisfied for
\emph{any} $y$; the sole surviving constraint is
$4x^2\ddot{y}=0$ from~\eqref{eq:conformal_Z}, with general integral
$y=ax+b$. Since $1-Z\equiv0$, equations~\eqref{eq:rho_Z}
and~\eqref{eq:p_Z} give $\rho=p=0$ for $N\ge2$. \qed
\end{proof}

At $Z\equiv1$ the first factor of~\eqref{eq:factorised} vanishes
identically while the second reduces to $4(1-N)x\dot{y}$; the whole
family $y=ax+b$ therefore lies on the singular branch, but only its
constant-lapse member $a=0$ annihilates the second factor as well and
so belongs to both. These solutions carry no matter and are of no
stellar interest, but we mention them for completeness.

\begin{corollary}[Two-branch structure]
\label{cor:twobranch}
The solution set of the system~\eqref{eq:isotropy_Z},
\eqref{eq:conformal_Z} is the union of the singular branch of
Theorem~\ref{thm:schwbranch} and the generic branch, the latter
comprising the non-constant family~\eqref{eq:implicit_abs} together
with the two equilibria of Theorem~\ref{thm:general}. The non-constant
generic family is disjoint from the singular branch for $N\ge2$:
substituting $Z=1+\beta x$ into~\eqref{eq:master} leaves the residual
$2k(2k-1)\beta^2x^2$, which vanishes only for $k=0$ (that is $d=2N$,
below the critical dimension), $k=\tfrac12$ (that is $d=2$), or
$\beta=0$. The excluded case $\beta=0$ is precisely the flat vacuum
$Z\equiv1$, which is the one solution the two branches share. For $N=1$
the generic branch is empty and the singular branch alone exhausts the
solution set.
\end{corollary}

Theorem~\ref{thm:schwbranch} relates the analysis directly with the universality programme of Dadhich and
collaborators~\cite{Dadhich2010b,Dadhich2016a,Dadhich2017}, in which a
Schwarzschild-type vacuum solution, isothermal fluid spheres and the
uniform-density interior all persist across Lovelock orders. The
theorem extends that persistence to conformal geometry: the
Schwarzschild interior remains conformally flat, with both metric
potentials retaining their Einstein-gravity functional form, at every
order $N$. What is genuinely new at higher order is the coexisting
generic family, which has no Einstein counterpart. In Einstein gravity
constant density and conformal flatness select the same solution,
unique up to the degenerate zero-density configuration of
Corollary~\ref{cor:einstein}; for $N\ge2$ they continue to intersect on
the Schwarzschild interior, but conformal flatness now admits in addition a
one-parameter continuum.

\subsection{Completeness}

The factorised identity~\eqref{eq:factorised} is an exact consequence
of the system~\eqref{eq:isotropy_Z},~\eqref{eq:conformal_Z}, so at
every point at least one of its two factors vanishes. That alone does
not preclude a solution which annihilates the first factor on part of
its domain and the second on the rest, and the following lemma closes
this gap.

\begin{lemma}[No branch switching]
\label{lem:noswitch}
A generic-branch trajectory cannot meet the singular locus
$A\equiv x\dot{Z}-Z+1=0$ at any finite $\tau=\ln x$. The two branches
are therefore separated, and each maximal solution lies entirely on one
of them.
\end{lemma}

\begin{proof}
On the generic branch $x\dot{Z}$ is given by the master
equation~\eqref{eq:master}, whence
\begin{equation}
  A = x\dot{Z}-Z+1
    = -\frac{2k(2k-1)(Z-1)^2}{2k+(2k-1)Z}.
\label{eq:Aformula}
\end{equation}
For $k<0$ the coefficient $2k(2k-1)$ is strictly positive, so $A$
vanishes only where $(Z-1)^2=0$, that is only at $Z=1$. But $Z=1$ is a
fixed point of the autonomous flow~\eqref{eq:autonomous}, and by
uniqueness of solutions no non-constant trajectory attains a fixed
point at finite $\tau$. Hence $A\neq0$ everywhere along a generic
trajectory. \qed
\end{proof}

On the singular branch Theorem~\ref{thm:schwbranch} determines the
solution completely. On the generic branch the reduction
to~\eqref{eq:master} involves no further loss of generality, and away
from the two fixed points the Picard--Lindel\"of theorem guarantees
uniqueness of trajectories, so the one-parameter
integral~\eqref{eq:implicit_abs} captures every non-constant
generic-branch solution, in either orientation.
With Lemma~\ref{lem:noswitch} the two-branch decomposition is
therefore complete, and what remains is the invertibility
of~\eqref{eq:implicit_abs}.

\subsection{Algebraic character and radical solvability}
\label{ssec:algebraic}

For every admissible pair $(d,N)$ the parameter $k$ of~\eqref{eq:k_def}
is rational, so the integral~\eqref{eq:implicit_abs} is an
\emph{algebraic} relation between $Z$ and $x$. Writing $2k=-u/v$ in
lowest terms and clearing fractional exponents gives, on each connected
real sector of the phase line,
\begin{equation}
  (Z-1)^{u+v}=c\,x^{u+v}\,L(Z)^{u},
\label{eq:algebraic}
\end{equation}
with $c$ a non-zero sector-dependent constant. This relation is
irreducible, so the spatial potential is an algebraic function of $x$
of degree exactly $u+v$ and the generic conformally flat pure Lovelock
metric is never transcendental in $x$
(Lemma~\ref{lem:irreducible}).

Inversion by radicals is guaranteed whenever $u+v\le4$. Enumerating the
coprime pairs gives the complete list of such families,
$k\in\{-\tfrac14,-\tfrac12,-1,-\tfrac32,-\tfrac16\}$, corresponding
respectively to $d=3N-1$, $d=4N-2$, $d=6N-4$, $d=8N-6$ and $3d=8N-2$,
the last requiring $N\equiv1\pmod 3$ and so first arising physically at
$(N,d)=(4,10)$. Degree at most four is sufficient but not necessary,
since a polynomial of higher degree may still have a solvable Galois
group; for four representative higher-degree configurations, however,
the obstruction is genuine, the Galois group being the full symmetric
group $S_5$ or $S_7$ (Proposition~\ref{prop:galois}). The exact models
of the solvable families are developed in a companion
paper~\cite{HansrajCompanion}; here we retain the implicit and
parametric descriptions, which treat all $(d,N)$ uniformly.
Table~\ref{tab:cases} records the algebraic data for representative
cases. Proofs of both statements, together with the modular
certificates, are collected in Appendix~\ref{app:galois}.

\begin{table}
\caption{Representative physical pairs $(d,N)$, the parameter $k$, the
algebraic degree $u+v$ of the relation~\eqref{eq:algebraic} where
$2k=-u/v$ in lowest terms, and the solvability status of the generic
branch}
\label{tab:cases}
\begin{ruledtabular}
\begin{tabular}{ccccl}
$N$ & $d$ & $k$ & $u+v$ & Status\\ \hline
2 & 5  & $-1/4$  & 3 & radicals (cubic)~\cite{HansrajCompanion}\\
2 & 6  & $-1/2$  & 2 & radicals (quadratic)~\cite{HansrajCompanion}\\
2 & 7  & $-3/4$  & 5 & insoluble, Galois group $S_5$\\
3 & 7  & $-1/8$  & 5 & insoluble, Galois group $S_5$\\
3 & 9  & $-3/8$  & 7 & insoluble, Galois group $S_7$\\
3 & 10 & $-1/2$  & 2 & radicals (quadratic)~\cite{HansrajCompanion}\\
4 & 9  & $-1/12$ & 7 & insoluble, Galois group $S_7$\\
\end{tabular}
\end{ruledtabular}
\end{table}

\begin{theorem}[Completeness and solvability]
\label{thm:nogo}
For $N\ge2$ and $d\ge2N+1$ the complete solution set of the conformally
flat isotropic pure Lovelock system consists of exactly the following,
and nothing else:
\begin{enumerate}
\item[(a)] the constant-density Schwarzschild interior
  $Z=1+\beta x$, $y=a\sqrt{1+\beta x}+b$ with $\beta\neq0$;
\item[(b)] the spatially flat vacuum $Z\equiv1$, $y=ax+b$, with
  $\rho=p=0$; its constant-lapse subcase $a=0$ is the sole
  intersection of the two branches, since at $Z\equiv1$ the second
  factor of~\eqref{eq:factorised} reduces to $4(1-N)x\dot{y}$ and so
  vanishes only when $\dot{y}=0$;
\item[(c)] the power-law fixed-point solution $Z\equiv Z^*$,
  $y=y_0x^{(4k-1)/(4k)}$ of~\eqref{eq:fixedsol}, which is the pure
  Lovelock isothermal sphere when $d\ge2N+2$ and a degenerate
  zero-density state when $d=2N+1$, by~\eqref{eq:fixedpointcases};
\item[(d)] the non-constant generic family~\eqref{eq:implicit_abs},
  which contains every non-constant regular-centre solution on the
  generic branch --- the singular branch of~(a) and the flat
  solution~(b) may also be regular at the centre, whereas~(c) has
  $p\propto x^{-N}$ and is singular there. This family has two
  orientations: the physically admissible one occupies the arc
  $Z^*<Z<1$, on which~\eqref{eq:implicit_abs} reduces to the real
  form~\eqref{eq:implicit}, while the orientation entering $Z>1$ is
  excluded by the conditions of Sect.~\ref{ssec:central}. Note that
  \eqref{eq:implicit} is not real on $Z>1$ for fractional $2k$, so it
  is~\eqref{eq:implicit_abs} and not~\eqref{eq:implicit} that is the
  complete integral.
\end{enumerate}
On (d), $Z(x)$ is an algebraic function of $x$ of degree exactly $u+v$,
where $2k=-u/v$ in lowest terms
(Lemma~\ref{lem:irreducible}). A closed form by radicals is guaranteed
whenever $u+v\le4$, that is for
$k\in\{-\tfrac14,-\tfrac12,-1,-\tfrac32,-\tfrac16\}$; for $u+v\ge5$ it
holds if and only if the Galois group of~\eqref{eq:algebraic} is
solvable, which fails for the four cases of
Proposition~\ref{prop:galois}.
\end{theorem}

\section{Phase-space structure of the generic branch}
\label{sec:dynamical}

Casting the master equation as a one-dimensional autonomous system
yields a topological proof that no generic-branch solution has a finite
pressure-free boundary, an exact analytic prediction
of the large-$r$ power-law decay, and the identification of the
attractor with a pure Lovelock isothermal sphere, the last for
$d\ge2N+2$; in the critical dimension $d=2N+1$ the attractor is instead
the zero-density state of~\eqref{eq:fixedpointcases}. Each result is
dimension-independent and holds for every admissible $N\ge2$,
$d\ge2N+1$. The
attractor statements concern the trajectories that converge on $Z^*$,
which include the whole physical arc; the complementary orbits, among
them the $a<0$ orientation of Sect.~\ref{ssec:central}, escape and are
treated separately.

\subsection{Autonomous system, fixed points and eigenvalues}

The substitution $\tau=\ln x$ transforms~\eqref{eq:master} into
\begin{equation}
  \frac{dZ}{d\tau}=f(Z)\equiv
  \frac{(Z-1)\bigl[4k^2-(2k-1)^2Z\bigr]}{2k+(2k-1)Z},
\label{eq:autonomous}
\end{equation}
a flow on the line. The physical domain $x\in(0,\infty)$ corresponds to
$\tau\in(-\infty,\infty)$, the spacetime origin $r=0$ lying at
$\tau\to-\infty$ and the asymptotic region $r\to\infty$ at
$\tau\to+\infty$. The entire qualitative behaviour of every solution is
therefore encoded in the sign of $f$ and the location of its zeros.

Setting $f(Z)=0$ gives exactly two fixed points,
\begin{equation}
  Z_1=1,\qquad Z_2=Z^*\equiv\frac{4k^2}{(2k-1)^2},
\label{eq:fixedpoints}
\end{equation}
the first being the locus of flat spatial sections and the second, as
shown in Sect.~\ref{ssec:isothermal}, an exact power-law solution,
which is the pure Lovelock isothermal sphere for $d\ge2N+2$ and a
degenerate zero-density state in the critical dimension $d=2N+1$.
Their stability is fixed by $\lambda=f'(Z)$, which evaluates in closed
form to
\begin{align}
  \lambda_1 &\equiv f'(1)
  = \frac{4k^2-(2k-1)^2}{2k+(2k-1)}=\frac{4k-1}{4k-1}=1>0,
\label{eq:lambda1}\\
  \lambda_* &\equiv f'(Z^*)
  = -\frac{2k-1}{2k}=-\frac{d-2}{\,d-2N\,}<0 .
\label{eq:lambdastar}
\end{align}
Hence $Z_1=1$ is an unstable repeller with a \emph{universal}
eigenvalue independent of $(d,N)$, while $Z^*$ is a stable attractor
with $0<Z^*<1$ for every physical pair with $N\ge2$. The closed form in
\eqref{eq:lambdastar} is worth emphasising: the settling rate onto the
asymptotic state is fixed by the elementary ratio $(d-2)/(d-2N)$.
Table~\ref{tab:fixedpoints} lists the data for the cases studied below.

\begin{table}
\caption{Fixed-point data for the six configurations studied. All cases
have $k<0$ and $Z^*\in(0,1)$; the eigenvalue $\lambda_1=1$ at $Z_1=1$
is universal}
\label{tab:fixedpoints}
\begin{ruledtabular}
\begin{tabular}{ccccc}
$N$ & $d$ & $k$ & $Z^*=4k^2/(2k-1)^2$ & $\lambda_*$\\ \hline
2 & 5  & $-1/4$ & $1/9\approx0.111$   & $-3$\\
2 & 6  & $-1/2$ & $1/4=0.250$         & $-2$\\
2 & 7  & $-3/4$ & $9/25=0.360$        & $-5/3$\\
3 & 7  & $-1/8$ & $1/25=0.040$        & $-5$\\
3 & 9  & $-3/8$ & $9/49\approx0.184$  & $-7/3$\\
3 & 10 & $-1/2$ & $1/4=0.250$         & $-2$\\
\end{tabular}
\end{ruledtabular}
\end{table}

\subsection{Compactification and the desingularised flow}
\label{ssec:compactification}

The right-hand side of~\eqref{eq:autonomous} is a rational function
whose denominator vanishes, and a complete account of the global orbit
structure must accommodate that pole. Writing
\begin{equation}
  f(Z)=\frac{P(Z)}{Q(Z)},\quad
  \begin{aligned}
    P(Z)&=-(2k-1)^2(Z-1)(Z-Z^*),\\
    Q(Z)&=2k+(2k-1)Z,
  \end{aligned}
\label{eq:PQ}
\end{equation}
the line carries three distinguished points rather than two: the fixed
points at which $P$ vanishes, and the pole
\begin{equation}
  Z^Q_\infty=-\frac{2k}{2k-1},
\label{eq:pole}
\end{equation}
at which $Q$ vanishes and $f$ diverges. For every physical
configuration $k<0$, so $Z^Q_\infty<0$ and the pole lies strictly below
the physical interval $(Z^*,1)$. The appropriate mathematical device is the
projective compactification of the line together with a time
reparametrisation that desingularises the pole; we emphasise that the
Poincar\'e--Lyapunov compactification onto a sphere is a device for
planar polynomial fields, whereas the present flow is genuinely
one-dimensional and its correct analogues are $\mathbb{RP}^1\cong S^1$
and the rescaling below.

Introducing a new evolution parameter $\eta$ through
$d\eta=d\tau/Q(Z)$ absorbs the denominator into the parametrisation and renders
the flow polynomial,
\begin{equation}
  \frac{dZ}{d\eta}=P(Z)=-(2k-1)^2(Z-1)(Z-Z^*),
\label{eq:desing}
\end{equation}
a logistic equation whose rest states are the two fixed points and in
which the pole has become an ordinary interior point, crossed with
finite speed $P(Z^Q_\infty)\neq0$. In $\eta$ the point $Z=1$ attracts
and $Z^*$ repels. The reparametrisation reverses orientation wherever
$Q<0$, and $Q<0$ on the whole of $(Z^Q_\infty,\infty)$, since
$Q(1)=4k-1<0$ and $Q$ is monotone; the physical stabilities of
Sect.~\ref{sec:dynamical} are therefore a consequence of the sign of
$Q$ rather than of the polynomial $P$, a point the bare phase line
leaves implicit.

The behaviour at the ends of the line is settled in the chart $W=1/Z$,
in which~\eqref{eq:desing} becomes
\begin{equation}
  \frac{dW}{d\eta}=(2k-1)^2\bigl[1-(1+Z^*)W+Z^*W^2\bigr],
\label{eq:infinity}
\end{equation}
whose value at $W=0$ is $(2k-1)^2\neq0$. This non-vanishing does more
than establish that the point at infinity is a regular point of the
desingularised flow: it fixes which compactification is the correct
one. Because the velocity there is non-zero, an orbit arriving along
$Z\to+\infty$ continues into $Z\to-\infty$, so the flow itself
identifies the two ends of the line and $\mathbb{RP}^1$ is forced
rather than chosen. Had $dW/d\eta$ vanished at $W=0$, infinity would
have been an equilibrium and the two-point compactification would have
been the appropriate object instead. We note in passing that the
Poincar\'e--Lyapunov sphere is not an alternative here but an
inapplicable one: that construction is defined by central projection
for polynomial fields on $\mathbb{R}^n$ with $n\ge2$, and at $n=1$ it
degenerates precisely to the projective compactification used above, so
the chart~\eqref{eq:infinity} already carries its entire content.

The compactified phase space is therefore a circle carrying exactly two
hyperbolic fixed points, which divide it into two arcs. One arc is the
physical interval $(Z^*,1)$, swept monotonically in $\tau$ from the
repeller $Z=1$ to the attractor $Z^*$ and carrying the entire
physically admissible family of regular interiors; the $a<0$
orientation of Sect.~\ref{ssec:central} is also regular at the centre
but lies on $Z>1$ and is excluded there. The complementary arc passes
through the regular point at infinity and through the pole, and carries
the unphysical orbits, fenced off from the physical arc by the two
fixed points.

Two features of the sign structure deserve emphasis, because they are
easily misread from the physical interval alone. First, since $Q<0$
throughout $(Z^Q_\infty,\infty)$, the attractor $Z^*$ draws in the
whole interval $(Z^Q_\infty,1)$ and not merely the physical part above
it: on $(Z^Q_\infty,Z^*)$ one has $f>0$, so trajectories there also
converge on $Z^*$, from below. It is the pole, not the attractor, that
separates the escaping orbits from the rest, and the orbits that run
off to infinity are those below $Z^Q_\infty$ and those above $Z=1$.
Second, near $Z^Q_\infty$ one has
$\dot{Z}\sim\text{const}/(Z-Z^Q_\infty)$, so that
$(Z-Z^Q_\infty)^2\sim\tau-\tau_0$: the trajectory reaches the pole with
a vertical tangent in \emph{finite} $\tau$, while $d\tau/d\eta=Q$
passes through zero and changes sign. We draw no conclusion about the
nature of that locus: it lies at $Z<0$, hence outside the Lorentzian
static sector described by~\eqref{eq:metric} as well as outside the
physical interval, and no trajectory on the physical arc reaches it.
Whether it is a curvature singularity or merely a breakdown of the
areal radial coordinate is not settled by the desingularised flow
alone, and is not needed for anything that follows.

The pole and the point at infinity are reached on quite different terms, and the distinction is worth recording since the two statements above are made in different parametrisations along the same trajectories. Near infinity $f(Z)\sim(1-2k)Z$
with $1-2k>0$, so $Z\sim e^{(1-2k)\tau}$ and the end of the line is
attained only as $\tau\to+\infty$; the crossing is completed in finite
$\eta$ solely because $d\tau/d\eta=Q$ diverges there. The pole, by contrast, is genuinely attained at finite $\tau$. In the original radial variable, then, no trajectory reaches infinity at finite areal radius, and the smooth passage through $W=0$ is a property of the desingularised parametrisation rather than of the areal one.

Figure~\ref{fig:phase} displays the phase portraits. Since $f(Z)<0$ for all $Z\in(Z^*,1)$, every physical trajectory moves strictly leftward on the phase line: it originates at the repeller $Z=1$, descends monotonically through the interior region in which $\rho>0$ and $p>0$, and converges on $Z^*$ exponentially in $\tau$, that is as a power law
in $r$.

\begin{figure}
  \centering
  \includegraphics[width=0.48\textwidth]{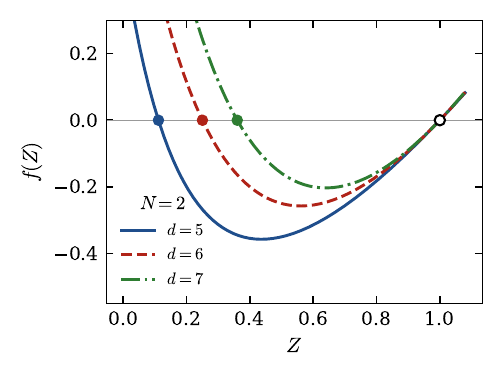}\hfill
  \includegraphics[width=0.48\textwidth]{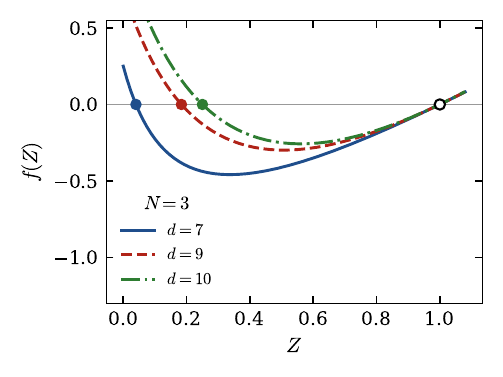}
  \caption{Phase portraits of the autonomous flow~\eqref{eq:autonomous}
  for the Gauss--Bonnet cases (left) and the cubic Lovelock cases
  (right). Filled circles mark the attractors $Z^*$; the open circle
  marks the universal repeller $Z=1$ with eigenvalue $\lambda_1=1$}
  \label{fig:phase}
\end{figure}

\subsection{\texorpdfstring{Large-$r$}{Large-r} decay law}

Near the attractor, write $Z=Z^*+\epsilon(\tau)$ with
$|\epsilon|\ll1$. Linearising~\eqref{eq:autonomous} gives
$d\epsilon/d\tau=\lambda_*\epsilon$, so that
$\epsilon\sim Ae^{\lambda_*\tau}=Ax^{\lambda_*}$ and, since $x=Cr^2$,
$\epsilon\sim A(Cr^2)^{\lambda_*}=\tilde{A}r^{2\lambda_*}$ with
$C^{\lambda_*}$ absorbed into $\tilde{A}$, so that
\begin{equation}
  Z(r)\approx Z^*+\tilde{A}\,r^{2\lambda_*}\qquad(r\to\infty).
\label{eq:Zasymptote}
\end{equation}
Substituting into the pressure formula~\eqref{eq:pressure_explicit},
\begin{equation}
  p(r)=K\left[\frac{(1-Z^*)^N}{r^{2N}}
  -\frac{N\tilde{A}(1-Z^*)^{N-1}}{r^{2N-2\lambda_*}}
  +\mathcal{O}\bigl(r^{4\lambda_*-2N}\bigr)\right],
\label{eq:pressuredecay}
\end{equation}
with $\tilde{A}$ as in~\eqref{eq:Zasymptote} and
\begin{equation}
  K=\frac{(d-2)!\,(d-N-1)}{2(d-2N-1)!\,(N-1)} .
\label{eq:Kdef}
\end{equation}
The factor $C^{N}$ carried by~\eqref{eq:pressure_explicit} cancels
against $x^{N}=C^{N}r^{2N}$, so no power of $C$ survives once the
result is written in terms of the areal radius: the decay law is a
statement about $r$ alone, independent of the normalisation
in~\eqref{eq:transform}. The leading decay $p\sim r^{-2N}$ is therefore \emph{universal} across
every trajectory converging on $Z^*$, and in particular across the
whole physical arc $Z^*<Z<1$; the $a<0$ orientation of
Sect.~\ref{ssec:central}, which escapes to $Z\to\infty$ instead, obeys
the different law~\eqref{eq:aneg}. The correction term carries the
$(d,N)$ dependence through $\lambda_*$. The eigenvalue thus plays a
double role: as the stability exponent of $Z^*$ in the phase flow, and
as the correction-to-leading exponent in the large-$r$ pressure
profile.

\subsection{Absence of a finite pressure-free boundary}

\begin{theorem}[No finite boundary]
\label{thm:unbounded}
Every generic-branch solution with $Z\not\equiv1$ and $N\ge2$ satisfies
$p(r)\neq0$ at every finite $r$; on the physical arc $Z^*<Z<1$ one has
$p(r)>0$ throughout, so no pressure-free boundary occurs at finite
radius. The Schwarzschild branch of Theorem~\ref{thm:schwbranch} is
consequently the only branch capable of describing a smoothly bounded
perfect-fluid star at any Lovelock order.
\end{theorem}

\begin{proof}
A bounded stellar object requires $p=0$ at some finite radius $r=R$.
By~\eqref{eq:pressure_explicit}, $p=0$ at finite $x$ if and only if
$Z=1$. But $Z=1$ is the repeller $Z_1$, a fixed point of the
flow~\eqref{eq:autonomous}, so by uniqueness no solution with
$Z\not\equiv1$ attains it at finite $\tau$; hence $p(r)\neq0$ at every
finite $r$ on any such solution. On the physical arc, all trajectories
in $(Z^*,1)$ have $f(Z)<0$, so $Z(\tau)$ is strictly decreasing and
bounded below by $Z^*>0$, giving $Z(r)<1$ and therefore $p(r)>0$ for all
finite $r$: the pressure vanishes only asymptotically. The excluded case
$Z\equiv1$ is the vacuum of Theorem~\ref{thm:schwbranch}, for which
$p\equiv0$ everywhere and no boundary exists either. \qed
\end{proof}

The phase-line picture clarifies the relation between the branches.
The singular branch $Z=1+\beta x$ is not a trajectory
of~\eqref{eq:autonomous} at all: it lies on the locus annihilating the
first factor of~\eqref{eq:factorised}, and it is precisely the family
that exits $Z=1$ and supports a pressure-free surface. On the generic
branch the repeller at $Z=1$, with universal eigenvalue $\lambda_1=1$,
prevents any trajectory from straddling both a stellar centre and a
stellar surface: boundedness and genericity are mutually exclusive. In
Einstein gravity the generic branch is empty and only the singular
branch survives, whose $\beta\neq0$ members are the bounded
Schwarzschild interiors; this is the classical uniqueness result seen
from the phase line.

\subsection{The attractor as an isothermal sphere}
\label{ssec:isothermal}

The fixed point $Z=Z^*$ is itself an exact solution. With $Z$ constant,
\eqref{eq:ydot_over_y} integrates to a pure power law,
\begin{equation}
  y\propto x^{m},\qquad m=\frac{k(Z^*-1)}{Z^*}=\frac{4k-1}{4k},
\end{equation}
which is the solution~\eqref{eq:fixedsol}. At a fixed point $\dot{Z}=0$,
so~\eqref{eq:rho_Z} reduces to
$\rho\propto(d-2N-1)(1-Z^*)^{N}x^{-N}$, while
$p\propto(d-N-1)(1-Z^*)^{N}x^{-N}/(N-1)$
from~\eqref{eq:pressure_explicit}. The factor $(d-2N-1)$ in the density
but not in the pressure makes the critical dimension exceptional, and
the two cases must be separated:
\begin{equation}
  \begin{aligned}
   d\ge2N+2:&\;\; Z^*\ \text{is the pure Lovelock isothermal sphere},\\
   d=2N+1:  &\;\; \rho(Z^*)=0,\ p(Z^*)\neq0 .
  \end{aligned}
\label{eq:fixedpointcases}
\end{equation}

For $d\ge2N+2$ both $\rho$ and $p$ scale as $x^{-N}$ with constant ratio
\begin{equation}
  w_\infty\equiv\frac{p}{\rho}=\frac{d-N-1}{(N-1)(d-2N-1)},
  \qquad d\ge2N+2 ,
\label{eq:winf}
\end{equation}
an isothermal fluid sphere of the type whose universality across
Lovelock orders was established in~\cite{Dadhich2016a}, with the
equation-of-state parameter now selected by conformal flatness. At
$N=1$, $d=4$ the density profile $\rho\propto r^{-2N}$ reduces to the
$r^{-2}$ law of the singular isothermal sphere, the standard
idealisation of an extended self-gravitating
halo~\cite{BinneyTremaine,Chandrasekhar1939};
equation~\eqref{eq:winf} is thus a higher-curvature isothermal-halo
analogue, with the Lovelock term steepening the fall-off from $r^{-2}$
to $r^{-2N}$. We stress that the analogy is structural rather than
phenomenological: for $N>1$ the density falls considerably faster than
the four-dimensional profile associated with flat rotation curves, and
no claim about galactic dynamics is intended. Every
generic trajectory in the basin of $Z^*$ --- in particular every
solution on the physical arc $Z^*<Z<1$ --- relaxes in these dimensions,
at the rate set by $\lambda_*$, onto this universal isothermal state:
the attractor of the phase flow \emph{is} the isothermal sphere. The
escaping orbits, among them the $a<0$ orientation of
Sect.~\ref{ssec:central}, do not approach it at all.

This gives the attractor a direct physical interpretation. The isothermal profile is not merely one admissible solution among the conformally flat family
but the state that the entire family approaches: whatever regular central data a configuration begins with, its exterior forgets them and
settles onto the same universal law, with only the overall scale
retained. In this sense conformal flatness selects the isothermal halo
as its generic large-$r$ behaviour, and the closed-form eigenvalue
$\lambda_*=-(d-2)/(d-2N)$ of~\eqref{eq:lambdastar} measures how quickly
that memory is lost.

In the critical dimension $d=2N+1$ the coefficient $(d-2N-1)$ vanishes
and the fixed point carries pressure but no energy density. It is
therefore \emph{not} an isothermal fluid, $w_\infty$
of~\eqref{eq:winf} is undefined there, and the attractor is a
degenerate critical state rather than a member of the isothermal
family. Trajectories still converge to $Z^*$ geometrically, but their
density is controlled by the leading perturbation away from the fixed
point rather than by the fixed point itself, which is the analytic
origin of the anomalous tail derived below and of the divergence of
$p/\rho$ observed numerically in Sect.~\ref{sec:physical}.

Two consequences follow at once. First, the asymptotic sound speed is
$dp/d\rho\to w_\infty$, so the far field is causal only when
$w_\infty\le1$, that is when
\begin{equation}
  (N-2)\,d \;\ge\; 2N^2-2N-2 .
\label{eq:causalitybound}
\end{equation}
No Gauss--Bonnet case satisfies this bound. At $N=3$ it reads
$d\ge10$, saturated exactly at $d=10$ and satisfied strictly for every
$d>10$; at $N=4$ it requires $d\ge11$. This explains analytically why
$(N,d)=(3,10)$ emerges in Sect.~\ref{sec:physical} as the only
subluminal configuration \emph{among the six plotted}: it is the
smallest admissible dimension meeting the bound at $N=3$, not a
globally isolated case. Whenever
$w_\infty>1$ the dominant energy condition, although satisfied at the
centre, fails beyond a finite crossing radius.

Second, in the critical dimension $d=2N+1$, where by
\eqref{eq:fixedpointcases} the fixed point carries no density, the
leading term of~\eqref{eq:rho_Z} vanishes identically and the density
is governed by the first correction, decaying as
$\rho\sim r^{2\lambda_*-2N}$, strictly faster than the pressure. The
ratio $p/\rho$ then grows without bound, consistent with $w_\infty$
being undefined there; the dominant condition fails in the far field,
the sound speed diverges, and the adiabatic index grows as
$\Gamma\sim r^{-2\lambda_*}$. These are the analytic signatures of the
critical-dimension tail exhibited numerically in
Sect.~\ref{sec:physical}.

\begin{remark}[The zero-density states]
\label{rem:zerodensity}
Two configurations in the classification carry pressure but no energy
density: the fixed-point solution~\eqref{eq:fixedsol} in the critical
dimension $d=2N+1$, by~\eqref{eq:fixedpointcases}, and the $N=1$,
$\beta=0$ member of Corollary~\ref{cor:einstein}. Neither should be
regarded as an admissible matter distribution. With $\rho=0$ and $p>0$
the null and strong conditions are satisfied, and the weak condition
only marginally, but the dominant condition $\rho\ge|p|$ fails, and
fails maximally: the equation-of-state ratio $p/\rho$ is unbounded, so
there is no frame in which the energy flux remains causal. A fluid with
pressure and no energy density is not a physically acceptable source.

The two cases are nevertheless retained, for different and equally
mundane reasons. The $N=1$ solution is a degenerate member of the
singular branch, recorded so that the Einstein-limit statement is
exhaustive; it carries no weight in the pure Lovelock analysis, where
the factor $(1-Z)^{N-1}$ removes it. The critical-dimension fixed point
matters more, because it is the endpoint of a genuine flow. It is
approached but never attained at finite radius, and what the vanishing
of $\rho(Z^*)$ signals is not the existence of an exotic source but the
breakdown of the far field in $d=2N+1$: the density is driven to zero
faster than the pressure, $p/\rho$ diverges, and the dominant condition
is violated beyond a finite radius. This is the analytic content of the
critical-dimension tail exhibited in Sect.~\ref{sec:physical}, and it
is properly read as a statement about the asymptotics of the critical
dimension rather than about the fixed point as a spacetime in its own
right. Consistently with this, no bound distribution of finite radius
exists in $d=2N+1$~\cite{Dadhich2017}.
\end{remark}

\subsection{Expansion about the repeller and central values}
\label{ssec:central}

Expanding~\eqref{eq:master} about the repeller with
$1-Z=ax+bx^2+\mathcal{O}(x^3)$ fixes the second-order coefficient
uniquely in terms of the first,
\begin{equation}
  b=\frac{2k(2k-1)}{4k-1}\,a^2,
\end{equation}
while $a\neq0$ remains free. Substituting into~\eqref{eq:rho_Z}
and~\eqref{eq:pressure_explicit} yields the finite central values
\begin{equation}
  \rho_c=\frac{(d-1)!\,(aC)^N}{2(d-2N-1)!},\quad
  p_c=\frac{(d-2)!\,(d-N-1)(aC)^N}{2(d-2N-1)!\,(N-1)},
\label{eq:centralvals}
\end{equation}
so that $p_c/\rho_c=(d-N-1)/[(N-1)(d-1)]<1$. For $a>0$, where $\rho_c$
and $p_c$ are positive, the dominant condition therefore always holds
at the centre; the sign restriction matters, since for odd $N$ the
choice $a<0$ makes both central quantities negative, as noted
immediately below.

\subsubsection*{The two regular-centre orientations}

Regularity at the centre does not by itself fix the sign of $a$, and
since $\rho_c$ and $p_c$ of~\eqref{eq:centralvals} depend on $a$ only
through $a^N$, both signs give positive central values whenever $N$ is
\emph{even}. There are therefore two distinct regular-centre orbits
leaving the repeller, and because the scaling action of
Lemma~\ref{lem:scaling} is by $a\mapsto a/\mu$ with $\mu>0$, which
cannot change the sign of $a$, they are
genuinely different orbits rather than rescalings of one another. The
case $a>0$ gives $Z<1$ and is the arc studied throughout; the case
$a<0$ gives $Z>1$. A concrete instance is $(N,d)=(2,6)$, where $a=-1$
reproduces exactly the central values $\rho_c=60$, $p_c=36$ of the
$a=+1$ solution.

The second orientation is nevertheless excluded by the physical
requirements, and it is worth deriving this rather than assuming it. For
$Z>1$ and $k<0$ one has $f(Z)>0$, so $Z$ increases outward and the
trajectory runs to $Z\to+\infty$; from~\eqref{eq:autonomous},
$f(Z)\sim(1-2k)Z$ as $Z\to\infty$, whence
\begin{equation}
  Z\sim x^{1-2k},\qquad
  p\sim\rho\sim\frac{Z^N}{x^N}\sim x^{-2Nk}=r^{-4Nk}.
\label{eq:aneg}
\end{equation}
Since $k<0$ the exponent $-4Nk$ is positive: pressure and density
\emph{grow} without bound with radius rather than decaying, so the
requirement $p\to0$ as $r\to\infty$ fails. For odd $N$ the orbit is
excluded more simply, since $a^N<0$ makes $\rho_c$ and $p_c$ negative at
the outset. In either parity the $a<0$ orientation is inadmissible, and
the condition $a>0$, with the physical arc $Z^*<Z<1$, is thereby
\emph{derived} rather than imposed. Direct numerical integration
confirms~\eqref{eq:aneg}: for $(N,d)=(2,6)$ with $a=-1$ the pressure
grows as $r^{+4}$, against $r^{-4}$ for $a=+1$.

The rejected orientation nevertheless admits a complete and rather
clean asymptotic characterisation, which is worth recording because it
locates precisely where the pathology lies. Since
\begin{equation}
  4Nk-2N-2k+1=(2N-1)(2k-1),
\label{eq:factorid}
\end{equation}
the M\"obius relation~\eqref{eq:ratio_closed} gives, as
$Z\to+\infty$,
\begin{equation}
  \frac{p}{\rho}\;\longrightarrow\;\frac{1}{2N-1},
  \qquad
  \frac{dp}{d\rho}\;\longrightarrow\;\frac{1}{2N-1},
\label{eq:aneg_eos}
\end{equation}
the second limit following because $p$ and $\rho$ carry the same
leading radial power in~\eqref{eq:aneg}. Both limits are independent
of the dimension. Moreover $Z$ \emph{increases} outward on this
orientation, so by Theorem~\ref{thm:monotone} the ratio $p/\rho$
\emph{decreases} outward, from its central value towards $1/(2N-1)$;
the descent is genuine since
\begin{equation}
  \frac{p_c}{\rho_c}-\frac{1}{2N-1}
  =\frac{N(d-2N)}{(N-1)(2N-1)(d-1)}>0
\end{equation}
for every admissible pair. The asymptotic equation of state of the rejected branch is therefore
not merely regular but subluminal and dimension-independent, and the
grounds for rejection differ with the parity of $N$. For \emph{even}
$N$, where the central matter variables are positive, the sole
obstruction is the unbounded outward growth of $\rho$ and $p$
in~\eqref{eq:aneg}, and certainly not the ratio between them. For odd
$N$ the orbit is excluded twice over, by the negative central values
$\rho_c,p_c<0$ as well as by that growth. The central sound speed and central
adiabatic index follow in closed form, independently of $a$ as
Lemma~\ref{lem:scaling} requires:
\begin{equation}
  \left.\frac{dp}{d\rho}\right|_c=\frac{d-N-1}{(N-1)(d+1)},\qquad
  \Gamma_c=\frac{N(d-2)}{(N-1)(d+1)}.
\label{eq:central}
\end{equation}
The central sound speed is subluminal for every $(d,N)$. For
completeness we record the corresponding adiabatic index
$\Gamma_c=N(d-2)/[(N-1)(d+1)]$; since the critical stability index for
$d$-dimensional pure Lovelock gravity is not presently
known~\cite{Chandrasekhar1964a,Chandrasekhar1964b}, we do not pursue
that quantity further. Causality violation, where it occurs, is a
far-field phenomenon and not a central one: for the non-critical
dimensions $d\ge2N+2$ it is governed by $w_\infty$
of~\eqref{eq:winf}, and in the critical dimension $d=2N+1$, where
$w_\infty$ does not exist, by the critical tail and the exponent
$\lambda_*$ of~\eqref{eq:lambdastar}.

\section{Parametric representation and thermodynamics}
\label{sec:parametric}

Whether or not~\eqref{eq:algebraic} can be inverted by radicals, all
physical quantities on the physically admissible generic orbit can be
rendered fully analytic by a parametric representation valid uniformly
in $k$. The
natural parameter is the projective coordinate adapted to the two fixed
points of Sect.~\ref{sec:dynamical}. On the physical domain
$Z^*<Z<1$ set
\begin{equation}
  s=\frac{(2k-1)^2Z-4k^2}{1-Z}=\frac{(2k-1)^2(Z-Z^*)}{1-Z},
\label{eq:s_def}
\end{equation}
which vanishes at $Z^*$, diverges at $Z=1$, and is strictly increasing
in $Z$ since $ds/dZ=(1-4k)/(1-Z)^2>0$. It therefore maps the physical
arc bijectively onto $s\in(0,\infty)$, with $s\to\infty$ at the centre
($Z\to1$, $x\to0$) and $s\to0^+$ at spatial infinity ($Z\to Z^*$,
$x\to\infty$); the direction of increasing $r$ is the direction of
decreasing $s$. Inversion gives
\begin{equation}
  Z(s)=\frac{4k^2+s}{(2k-1)^2+s},\qquad
  1-Z(s)=\frac{1-4k}{(2k-1)^2+s},
\label{eq:Zs}
\end{equation}
and substitution into~\eqref{eq:implicit} yields
\begin{equation}
  x(s)^{\,2k-1}=\frac{b\bigl[(2k-1)^2+s\bigr]}{(1-4k)\,s^{2k}},
\label{eq:xs}
\end{equation}
with $b>0$ the redefined constant of the family. Since $s>0$, all
fractional powers are real.

In this coordinate the flow~\eqref{eq:autonomous} reads
\begin{equation}
  \frac{ds}{d\tau}=-\frac{s\bigl[s+(2k-1)^2\bigr]}{s+2k(2k-1)},
\label{eq:sflow}
\end{equation}
whose linearisations at the two ends return $ds/d\tau\approx\lambda_*s$
as $s\to0$ and, in $u=1/s$, $du/d\tau\approx\lambda_1u=u$ as
$s\to\infty$, recovering the eigenvalues~\eqref{eq:lambda1}
and~\eqref{eq:lambdastar} without further computation. The pole
of~\eqref{eq:sflow} sits at $s=-2k(2k-1)<0$, again outside the physical
range, in agreement with Sect.~\ref{ssec:compactification}.

\subsection{Metric potentials and thermodynamic quantities}

Equations~\eqref{eq:Zs} and~\eqref{eq:xs} give the spatial potential.
The temporal potential follows by integrating~\eqref{eq:ydot_over_y}
along the trajectory and --- this is the point at which the
representation becomes a solution rather than a reduction --- the
integral evaluates in elementary closed form.

\begin{theorem}[Explicit parametric solution]
\label{thm:explicit}
For every $k$ in the admissible range~\eqref{eq:domain} the physical
non-constant orbit $Z^*<Z<1$ is given in fully explicit parametric
form, with $s\in(0,\infty)$ decreasing outward, by
\begin{align}
  Z(s) &= \frac{4k^2+s}{(2k-1)^2+s},
\label{eq:Zexp}\\
  x(s) &= \left[\frac{b\bigl[(2k-1)^2+s\bigr]}{(1-4k)\,s^{2k}}
          \right]^{\frac{1}{2k-1}},
\label{eq:xexp}\\
  y(s) &= y_0\,
          \bigl[s+(2k-1)^2\bigr]^{\frac{k}{2k-1}}
          \bigl(s+4k^2\bigr)^{1/2}
          s^{-\frac{4k-1}{2(2k-1)}},
\label{eq:yexp}
\end{align}
with $b>0$ and $y_0>0$ constants. No quadrature remains: both metric
potentials of every physically admissible regular generic-branch
interior are elementary functions of $s$, at arbitrary order $N$ and in
arbitrary admissible dimension $d$. The $a<0$ orientation of
Sect.~\ref{ssec:central} is not represented by $s>0$ and is not
covered. The restriction to $Z^*<Z<1$ is
what makes $s>0$, and hence all the fractional powers real; other
generic trajectories lie in $s<0$ and require separate branch choices.
\end{theorem}

\begin{proof}
Equations~\eqref{eq:Zexp} and~\eqref{eq:xexp} restate~\eqref{eq:Zs}
and~\eqref{eq:xs}. Differentiating~\eqref{eq:xexp} gives
\begin{equation}
  \frac{dx}{ds}=\frac{x(s)\bigl[(1-2k)s-2k(2k-1)^2\bigr]}
                     {(2k-1)\,s\bigl[(2k-1)^2+s\bigr]}<0 ,
\label{eq:dxds}
\end{equation}
and substituting this together with~\eqref{eq:Zexp}
into~\eqref{eq:ydot_over_y}, written as
$d\ln y/ds=k\,[Z(s)-1]\,[x(s)Z(s)]^{-1}dx/ds$, yields after partial
fractions the elementary decomposition
\begin{equation}
  \frac{d\ln y}{ds}
  =\frac{k}{(2k-1)\bigl[s+(2k-1)^2\bigr]}
   +\frac{1}{2\bigl(s+4k^2\bigr)}
   -\frac{4k-1}{2(2k-1)\,s}.
\label{eq:dlny}
\end{equation}
Each term integrates to a logarithm; exponentiating
gives~\eqref{eq:yexp}. \qed
\end{proof}

Two features are worth recording. The middle factor
of~\eqref{eq:yexp} is
$(s+4k^2)^{1/2}=[Z(s)]^{1/2}[(2k-1)^2+s]^{1/2}$, so the temporal
potential inherits the square-root structure of the $\cosh$
representation~\eqref{eq:y_integral} with no transcendental residue.
And substituting the triple~\eqref{eq:Zexp}--\eqref{eq:yexp} back into
the conformal flatness condition~\eqref{eq:conformal_Z} returns zero
identically, which we have verified symbolically for representative
$k$: the parametric form is an exact solution, not an asymptotic one.
The pressure~\eqref{eq:pressure_explicit} becomes
\begin{equation}
  p(s)=\frac{C^N(d-2)!\,(d-N-1)}{2(d-2N-1)!\,(N-1)}
  \cdot\frac{(1-4k)^N}{\bigl[(2k-1)^2+s\bigr]^N x(s)^N},
\end{equation}
and, using $dZ/ds=(1-4k)/[(2k-1)^2+s]^2$, the radial gradient of the
spatial potential is
\begin{equation}
  \dot{Z}=\frac{dZ/ds}{dx/ds}
  =\frac{(1-4k)(2k-1)\,s}
        {x(s)\bigl[(2k-1)^2+s\bigr]\bigl[(1-2k)s-2k(2k-1)^2\bigr]},
\end{equation}
fully explicit in $s$ and enabling direct evaluation of the
density~\eqref{eq:rho_Z}. The causality condition $0\le dp/d\rho\le1$
may then be checked analytically through
$dp/d\rho=(dp/ds)/(d\rho/ds)$, both derivatives being available in
closed form. Along the trajectory the pressure decreases monotonically
from its finite central value~\eqref{eq:centralvals}, attained as
$s\to\infty$ where the ratio $(1-Z)/x$ remains finite, to zero as
$s\to0^+$, in accordance with Theorem~\ref{thm:unbounded}.

\subsection{Monotonicity of the equation of state}
\label{ssec:monotone}

The ratio $p/\rho$ is likewise available in closed form, and its
monotonicity can be established analytically rather than inferred from
the numerical profiles.

\begin{theorem}[Monotone equation of state]
\label{thm:monotone}
On the generic branch, in either orientation,
\begin{equation}
  \frac{p}{\rho}
  =\frac{(2k-1)Z+2k}{\bigl(4Nk-2N-2k+1\bigr)Z-2k},
\label{eq:ratio_closed}
\end{equation}
and consequently
\begin{equation}
  \frac{d}{dZ}\!\left(\frac{p}{\rho}\right)
  =-\frac{4Nk(2k-1)}
        {\bigl[\bigl(4Nk-2N-2k+1\bigr)Z-2k\bigr]^{2}}<0 ,
\label{eq:ratio_deriv}
\end{equation}
since $k<0$ and $2k-1<0$. As $Z$ decreases monotonically outward along
every physical trajectory, $p/\rho$ increases monotonically with $r$,
rising from the central value $(d-N-1)/[(N-1)(d-1)]$ to the limit
\begin{equation}
  \frac{p}{\rho}\;\longrightarrow\;
  \begin{cases}
    w_\infty \text{ of~\eqref{eq:winf}}, & d\ge2N+2,\\[1mm]
    +\infty,                             & d=2N+1,
  \end{cases}
\label{eq:ratiolimit}
\end{equation}
the second case reflecting the vanishing of $\rho$ at the fixed point
in the critical dimension, by~\eqref{eq:fixedpointcases}.
\end{theorem}

\begin{proof}
Eliminating $x\dot{Z}$ from~\eqref{eq:rho_Z} by means of the master
equation~\eqref{eq:master} and dividing
by~\eqref{eq:pressure_explicit}, the factors $(1-Z)^{N-1}x^{-N}$
cancel; substituting $d=2N+4k(1-N)$ from~\eqref{eq:k_def} then
gives~\eqref{eq:ratio_closed}. Differentiating this M\"obius function
of $Z$ gives~\eqref{eq:ratio_deriv}, whose numerator $-4Nk(2k-1)$ is
negative because $-4Nk>0$ and $2k-1<0$. Monotonicity in $r$ follows
from $f(Z)<0$ on $(Z^*,1)$. \qed
\end{proof}

Theorem~\ref{thm:monotone} promotes a numerical observation to a
structural one, and shows that the central and limiting values bracket
the entire profile: the dominant energy condition, if it fails at all,
fails once and never recovers, so the crossing radius is unique. In the
critical dimension the ratio is unbounded, so the failure is
guaranteed.

\section{Physical analysis}
\label{sec:physical}

\subsection{Physical requirements}
\label{ssec:requirements}

We adopt the standard plausibility requirements of relativistic
astrophysics, adapted to the unbounded character established by
Theorem~\ref{thm:unbounded}. The density and pressure are required
to be positive and finite, and the metric functions $e^{\nu}$ and
$e^{\lambda}$ positive and non-singular, throughout. Causality demands
$0\le dp/d\rho\le1$. For the energy conditions we monitor the null
combination $\rho+p$, the dominant condition $\rho-p\ge0$
(equivalently $p/\rho\le1$ given $\rho>0$), and the strong condition on
the matter, which for a $d$-dimensional perfect fluid is the pair
$\rho+p\ge0$ and
\begin{equation}
  (d-3)\rho+(d-1)p\ge0 ,
\label{eq:sec}
\end{equation}
reducing to the familiar $\rho+3p\ge0$ at $d=4$.
Condition~\eqref{eq:sec} is algebraic in $T_{ab}$ and therefore remains
meaningful in Lovelock gravity; what does \emph{not} carry over is its
equivalence, valid under the Einstein equations, to the Ricci
convergence condition $R_{ab}u^au^b\ge0$, since the Lovelock field
equations relate $T_{ab}$ to $G^{(N)}_{AB}$ rather than to the Einstein
tensor. We therefore report~\eqref{eq:sec} as a condition on the source
and draw no geodesic-focusing conclusion from it. The pressure must
fall to zero as $r\to\infty$, consistent with the absence of a finite
boundary. The Buchdahl mass-to-radius bound~\cite{Buchdahl1959} applies
to bounded fluid distributions and is therefore relevant only to the
Schwarzschild interior branch.

\subsection{Numerical method}
\label{ssec:numerical}

We integrate the master equation~\eqref{eq:master} in
\textsc{Mathematica} via \texttt{NDSolve} with the
\texttt{StiffnessSwitching} option, setting $C=1$ so that $x=r^2$ with
$r$ the areal radial coordinate; all plots are presented against $r$.
Initial conditions are set at $x_0=10^{-3}$ with $Z(x_0)=1-ax_0$ and
$a=1$. By Lemma~\ref{lem:scaling} this single choice covers the whole
physical family: every other $a>0$ gives the same curve read in a
rescaled radial coordinate. The alternative orientation $a<0$ leaves
the physical arc and is excluded in Sect.~\ref{ssec:central}.

All thermodynamic quantities are evaluated algebraically from the state
variables rather than by differentiating the interpolating function:
the first derivative is taken from the flow itself, $\dot{Z}=f(Z)/x$,
and the second by differentiating the master equation, so that
$dp/dx$, $d\rho/dx$, the sound speed $v^2=(dp/dx)/(d\rho/dx)$ and the
adiabatic index are exact algebraic expressions in $(x,Z)$. This
eliminates interpolation noise and guarantees that the features
reported below are properties of the solutions and not of the mesh. The
implicit integral~\eqref{eq:implicit} is verified by confirming
constancy of its left-hand side to relative error below $10^{-8}$ along
each solution, and for $k=-\tfrac12$ the algebraic
relation~\eqref{eq:algebraic} is quadratic in $Z$ and provides an
independent exact check of the $(N,d)=(2,6)$ and $(3,10)$ profiles. The
central values reproduce the closed forms of
Sect.~\ref{ssec:central}; for instance $\rho_c=60$ and $p_c=36$ for
$(N,d)=(2,6)$ with $a=1$.

\subsection{Gauss--Bonnet cases \texorpdfstring{($N=2$)}{(N=2)}}
\label{ssec:GB}

We present results for $(N=2,\;d=5,6,7)$, corresponding to
$k=-\tfrac14,-\tfrac12,-\tfrac34$ and $Z^*=1/9,\,1/4,\,9/25$, the
first of these being the critical case $d=2N+1=5$, whose distinctive
far field was derived in Sect.~\ref{ssec:isothermal}.

Figure~\ref{fig:attract} displays the approach of the spatial potential
to the attractor values: the
ordering of the plateaux follows the analytic values of $Z^*$ and the
settling rates follow the eigenvalues $\lambda_*=-3,-2,-\tfrac53$.
Figure~\ref{fig:GB_profiles} shows density and pressure on logarithmic
axes together with the universal $r^{-4}$ asymptote
of~\eqref{eq:pressuredecay}: all three pressure profiles settle onto
the same power law, the central values follow the closed forms of
Sect.~\ref{ssec:central}, and the density of the critical case $d=5$
visibly steepens onto the faster decay
$\rho\sim r^{2\lambda_*-2N}=r^{-10}$.

\begin{figure}
  \centering
  \includegraphics[width=0.48\textwidth]{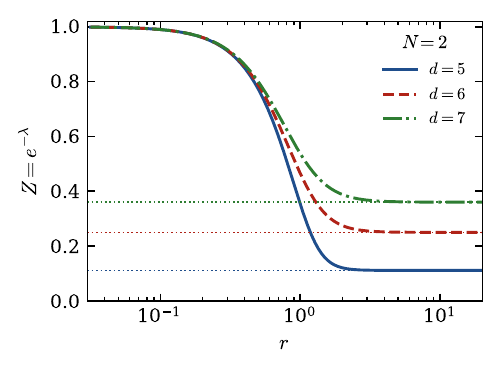}\hfill
  \includegraphics[width=0.48\textwidth]{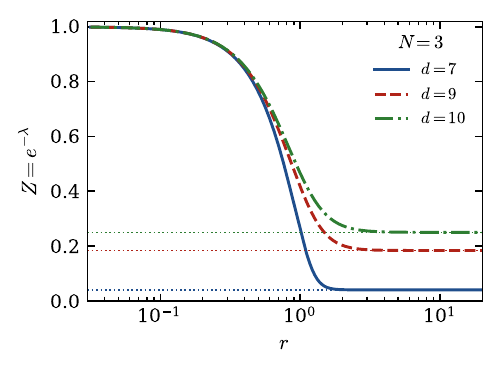}
  \caption{Approach of the spatial potential $Z=e^{-\lambda}$ to the
  attractor values $Z^*$ (dotted lines) for the Gauss--Bonnet cases
  (left) and the cubic Lovelock cases (right)}
  \label{fig:attract}
\end{figure}

\begin{figure}
  \centering
  \includegraphics[width=0.48\textwidth]{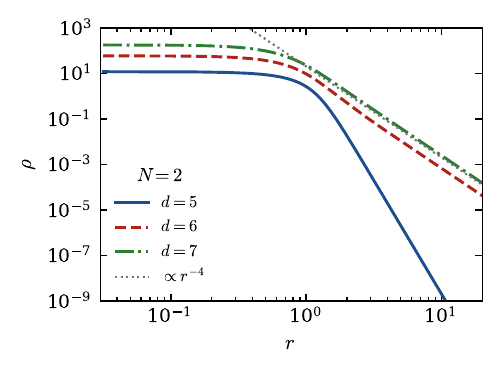}\hfill
  \includegraphics[width=0.48\textwidth]{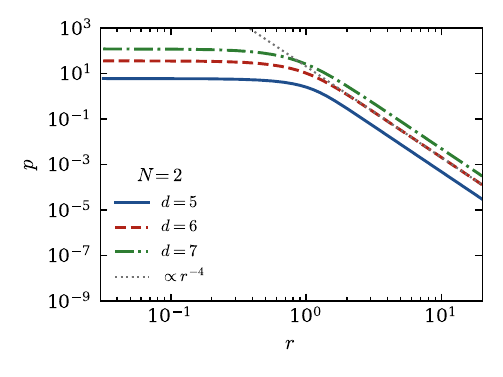}
  \caption{Energy density (left) and pressure (right) for the
  Gauss--Bonnet cases on logarithmic axes. The dotted guide marks the
  universal $r^{-2N}=r^{-4}$ decay of~\eqref{eq:pressuredecay}}
  \label{fig:GB_profiles}
\end{figure}

\subsection{Cubic Lovelock cases \texorpdfstring{($N=3$)}{(N=3)}}
\label{ssec:CL}

We present results for $(N=3,\;d=7,9,10)$, corresponding to
$k=-\tfrac18,-\tfrac38,-\tfrac12$ and $Z^*=1/25,\,9/49,\,1/4$. The
critical case $d=7$ exhibits the smallest fixed point of all
configurations studied, $e^{-\lambda}\to0.04$ asymptotically,
reflecting the strongest geometric deformation, and its eigenvalue
$\lambda_*=-5$ is the most negative of the six, producing the most
rapid settling onto the asymptote. The profiles are shown in
Fig.~\ref{fig:CL_profiles}: the pressure decays universally as
$r^{-6}$, while the density of the critical case steepens onto
$\rho\sim r^{2\lambda_*-2N}=r^{-16}$, an asymptotic slope approached
only slowly on the plotted range.

\begin{figure}
  \centering
  \includegraphics[width=0.48\textwidth]{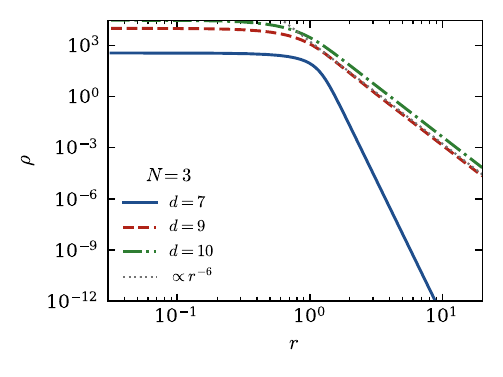}\hfill
  \includegraphics[width=0.48\textwidth]{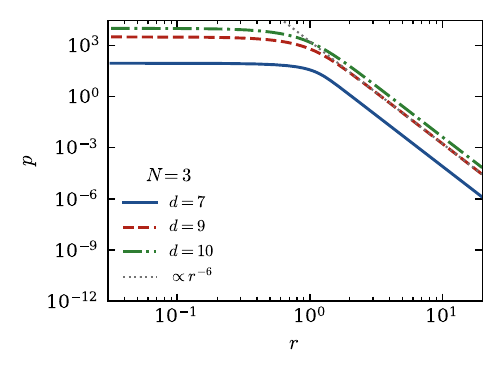}
  \caption{Energy density (left) and pressure (right) for the cubic
  Lovelock cases on logarithmic axes. The dotted guide marks the
  universal $r^{-2N}=r^{-6}$ decay}
  \label{fig:CL_profiles}
\end{figure}

\subsection{Causality, energy conditions and stability}
\label{ssec:causality}

Figure~\ref{fig:ratio} displays the ratio $p/\rho$, which encodes the
pointwise energy conditions at a glance: $\rho$ and $p$ are positive
throughout, so the null and weak conditions hold everywhere, and the
dominant condition holds precisely where $p/\rho\le1$. In every case the ratio rises monotonically from its central value
$(d-N-1)/[(N-1)(d-1)]$, as Theorem~\ref{thm:monotone} requires, the
curves illustrating that result rather than serving as evidence for it.
For the non-critical cases $d\ge2N+2$ the limit is the isothermal value
$w_\infty$ of~\eqref{eq:winf}, shown dashed; in the critical dimensions
$d=2N+1$, where $w_\infty$ does not exist, the ratio instead diverges,
by~\eqref{eq:ratiolimit}. For the Gauss--Bonnet cases
$w_\infty=3$ ($d=6$) and $w_\infty=2$ ($d=7$), and for $(N,d)=(3,9)$
one has $w_\infty=5/4$, so the dominant condition fails beyond a
crossing radius of order unity; in the
critical dimensions $d=2N+1$ the ratio grows without bound. The single
exception among the six is $(N,d)=(3,10)$, which saturates the
causality bound~\eqref{eq:causalitybound} with $w_\infty=1$: the ratio
approaches unity from below and the dominant condition is preserved,
marginally, at all radii. By~\eqref{eq:causalitybound} every $N=3$
configuration with $d>10$ behaves likewise, with $w_\infty<1$
strictly.

The sound speed (Fig.~\ref{fig:v2}) behaves analogously: it starts at
the subluminal central value $(d-N-1)/[(N-1)(d+1)]$
of~\eqref{eq:central} and increases towards $w_\infty$, so causality is
violated in the \emph{far} field rather than near the origin whenever
$w_\infty>1$, and diverges in the critical dimensions. The
configuration $(3,10)$ is again the only subluminal case among the six,
its sound speed approaching unity only asymptotically.

\subsection{The attractor displayed}
\label{ssec:attractorfig}

Figure~\ref{fig:attractorfig} exhibits the attractor directly rather
than through the phase portrait. The left panels plot
$r^{2N}\rho(r)$, which by~\eqref{eq:winf} must approach the constant
$(d-2N-1)(1-Z^*)^{N}$ times the prefactor of~\eqref{eq:rho_Z} whenever
$d\ge2N+2$; the curves flatten onto exactly those values, so the
numerical solutions are seen to settle onto the analytic isothermal law
$\rho\to\rho_0r^{-2N}$ and not merely to resemble it. The right panels
test the rate. Since $|Z-Z^*|\sim|\tilde{A}|r^{2\lambda_*}$
by~\eqref{eq:Zasymptote}, a logarithmic plot must be a straight line of
slope $2\lambda_*$, and the measured slopes reproduce
$2\lambda_*=-(2d-4)/(d-2N)$ to within a few parts in $10^{3}$ across
all six configurations: $-5.999$ against $-6$ for $(2,5)$, $-3.998$
against $-4$ for $(2,6)$, $-3.330$ against $-10/3$ for $(2,7)$,
$-9.971$ against $-10$ for $(3,7)$, $-4.666$ against $-14/3$ for
$(3,9)$ and $-3.998$ against $-4$ for $(3,10)$. The critical cases are
included in the right-hand panels, where the eigenvalue governs the
approach of $Z$ irrespective of the anomalous density tail, but omitted
from the left-hand panels, where $\rho(Z^*)=0$
by~\eqref{eq:fixedpointcases} and no finite plateau exists.

\begin{figure}
  \centering
  \includegraphics[width=\textwidth]{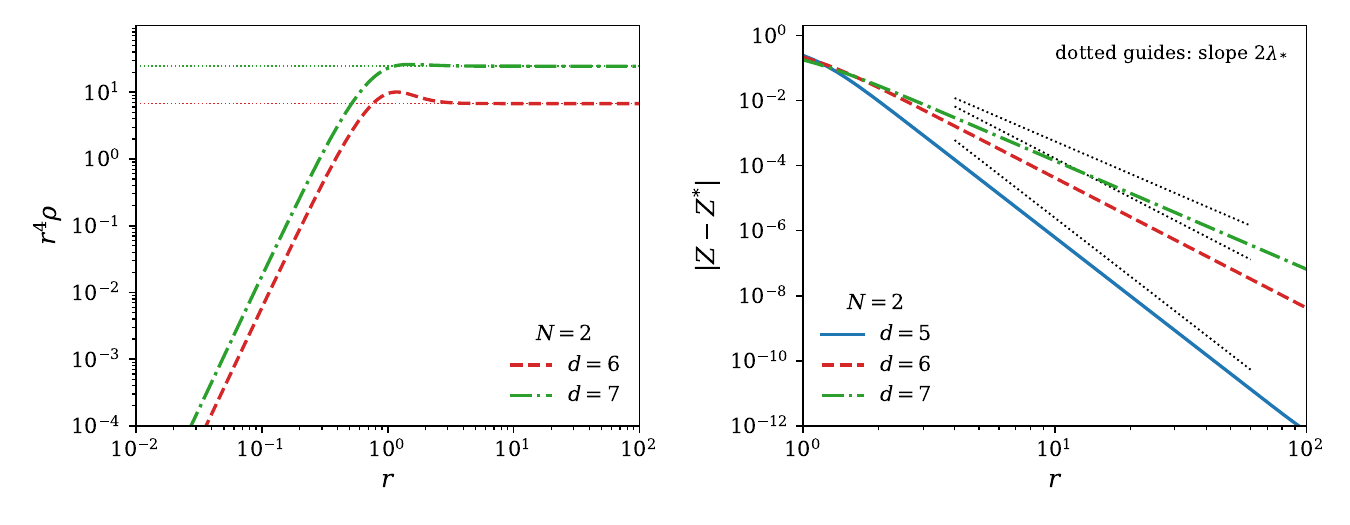}\\[2ex]
  \includegraphics[width=\textwidth]{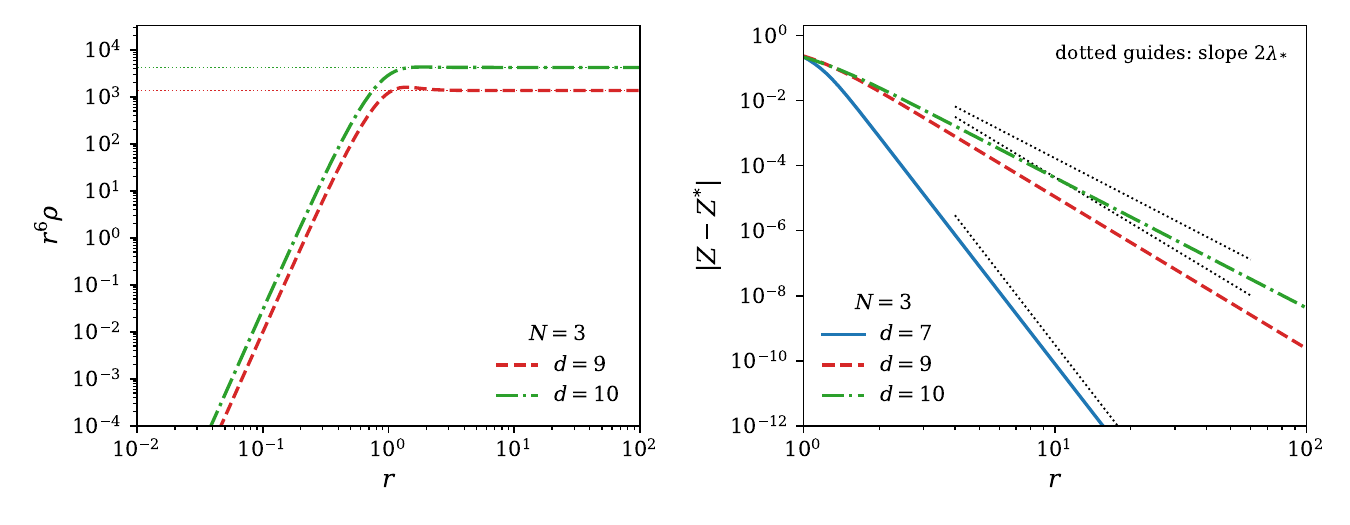}
  \caption{The isothermal attractor, for the Gauss--Bonnet
  ($N=2$, upper) and cubic Lovelock ($N=3$, lower) families. Left:
  $r^{2N}\rho$ approaches the constant predicted
  by~\eqref{eq:winf} (dotted), so that $\rho\to\rho_0r^{-2N}$. Right:
  $|Z-Z^*|$ on logarithmic axes, with dotted guides of the predicted
  slope $2\lambda_*$; the critical cases $d=2N+1$ appear here but not
  on the left, where the limiting density vanishes}
  \label{fig:attractorfig}
\end{figure}

\begin{figure}
  \centering
  \includegraphics[width=0.48\textwidth]{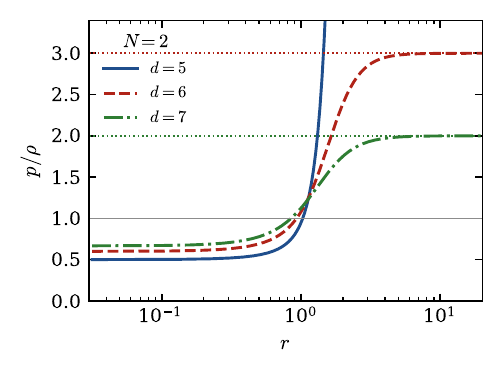}\hfill
  \includegraphics[width=0.48\textwidth]{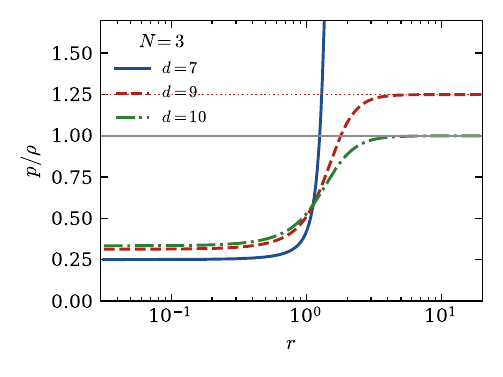}
  \caption{The ratio $p/\rho$ for the Gauss--Bonnet (left) and cubic
  Lovelock (right) cases. Dashed lines mark the isothermal values
  $w_\infty$ of~\eqref{eq:winf}; the grey line marks the dominant
  energy condition boundary $p/\rho=1$}
  \label{fig:ratio}
\end{figure}

\begin{figure}
  \centering
  \includegraphics[width=0.48\textwidth]{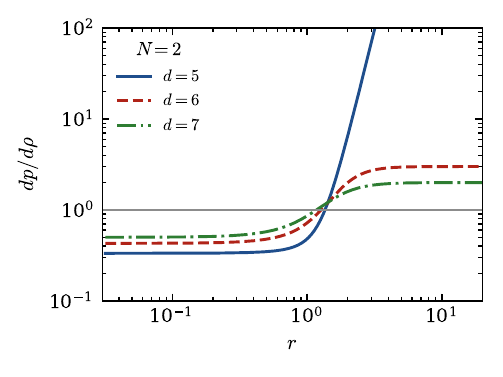}\hfill
  \includegraphics[width=0.48\textwidth]{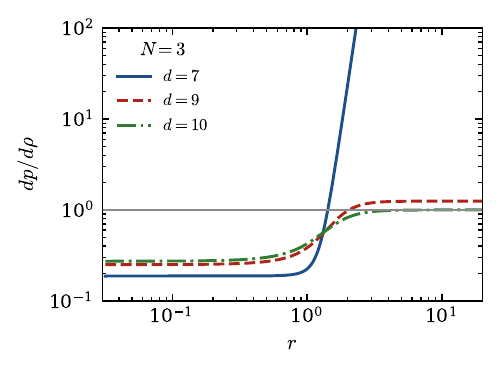}
  \caption{Sound speed $dp/d\rho$ for the Gauss--Bonnet (left) and
  cubic Lovelock (right) cases. The grey line marks the causal boundary
  $dp/d\rho=1$; central values are the closed forms
  of~\eqref{eq:central}}
  \label{fig:v2}
\end{figure}

\subsection{General observations}
\label{ssec:observations}

Across all six cases the numerical profiles confirm the analytic
structure of Sect.~\ref{sec:dynamical}. No pressure-free boundary
occurs at finite $r$ in any case, consistent with
Theorem~\ref{thm:unbounded}, and the pressure decays as $r^{-2N}$
with the correction exponent $2\lambda_*$ governing the settling rate.
The density is positive and monotonically decreasing from the finite
central value~\eqref{eq:centralvals}, with the critical dimensions
$d=2N+1$ exhibiting both the higher central densities and the anomalous
fast tail $\rho\sim r^{2\lambda_*-2N}$. Causality holds in a central
region of every configuration and is lost in the far field whenever
$w_\infty>1$, which occurs for every case plotted except
$(N,d)=(3,10)$; the same dichotomy governs the dominant energy
condition, while the null and weak conditions hold everywhere.
At fixed $N$, increasing $d$ raises $Z^*$ and slows the settling onto
the power law; at fixed $d$, increasing $N$ reduces $Z^*$, driving
$e^{-\lambda}$ to smaller asymptotic values. Finally, the implicit
integral~\eqref{eq:implicit} is confirmed constant along every
numerical solution to relative error below $10^{-8}$, and the
$k=-\tfrac12$ profiles agree with the exact quadratic
relation~\eqref{eq:algebraic} to comparable accuracy, validating
Theorem~\ref{thm:general}.

\section{Discussion}
\label{sec:discussion}

The central result is the two-branch structure of
Corollary~\ref{cor:twobranch}: the conformally flat isotropic pure
Lovelock spacetimes, for every $N\ge2$ and every admissible
$d\ge2N+1$, consist
of the constant-density Schwarzschild interior together with a generic
one-parameter family whose complete integral is~\eqref{eq:implicit_abs},
governed by the single dimension--order parameter $k=(d-2N)/[4(1-N)]$.
The parameter $k$ absorbs the entire $(d,N)$ dependence of the generic
branch. By Lemma~\ref{lem:scaling} the regular non-constant family in
$Z^*<Z<1$ carries a universal dimensionless profile, the remaining
constant setting only the radial and density scales; the map is a
homothety rather than an isometry, so its members are not isometric
copies of one another, and the equilibria lie outside the orbit
altogether.

Apart from the degenerate zero-density configuration identified in
Corollary~\ref{cor:einstein}, Einstein gravity selects the
Schwarzschild interior uniquely, constant density and conformal
flatness picking out the same solution.
Theorem~\ref{thm:schwbranch} shows that their intersection survives
intact at every Lovelock order, so that the
persistence established by Dadhich and
collaborators~\cite{Dadhich2010b,Dadhich2016a,Dadhich2017} for the
Schwarzschild vacuum, the isothermal sphere and the uniform-density
interior extends to conformal geometry. What distinguishes $N\ge2$ is
not the loss of the classical solution but the appearance alongside it
of a continuum of unbounded static configurations with no Einstein
counterpart; and by
Theorem~\ref{thm:unbounded} the classical solution remains the
unique \emph{bounded} member of the class. The Einstein uniqueness
theorem therefore generalises in the sharpened form: bounded and
conformally flat implies Schwarzschild interior, at every Lovelock
order. Relative to the Einstein--Gauss--Bonnet analysis
of~\cite{Hansraj2021}, which exhibited the second conformally flat
metric at $N=2$ in $d=5,6$, the present treatment supplies the
arbitrary-order factorisation~\eqref{eq:factorised}, the global
two-branch theorem for every admissible $(d,N)$, the explicit
parametric metric of Theorem~\ref{thm:explicit}, and the absence of a
finite pressure-free boundary on the second branch.

Theorem~\ref{thm:nogo} settles the closed-form question. Because $k$ is
rational for every admissible $(d,N)$, the spatial potential is always
an algebraic function of $x$, of degree exactly $u+v$ where $2k=-u/v$;
inversion by radicals is guaranteed on the families of degree at most
four, namely $d=3N-1$, $4N-2$, $6N-4$, $8N-6$ and $3d=8N-2$, and the
corresponding exact models are developed in the companion
paper~\cite{HansrajCompanion}. Degree five or more may still be
solvable in individual cases, so this list is where radicals are
assured rather than where they are possible; for the four cases
examined the obstruction is Galois-theoretic rather than merely
practical, their Galois groups being the full symmetric group
(Proposition~\ref{prop:galois}). The parametric
representation of Sect.~\ref{sec:parametric} then supplies the analytic
description of the physical generic orbit uniformly in $k$.

The phase-space analysis does substantially more than prove
Theorem~\ref{thm:unbounded}. The stability eigenvalue at the
attractor, obtained in the closed form
$\lambda_*=-(d-2)/(d-2N)$, is simultaneously the
correction-to-leading exponent of the large-$r$ pressure
profile~\eqref{eq:pressuredecay}: the leading decay $p\sim r^{-2N}$ is
universal and $\lambda_*$ controls how rapidly each trajectory in the
basin of $Z^*$ settles
onto it. The identification of the attractor with the pure Lovelock
isothermal sphere, valid for $d\ge2N+2$ and replaced in the critical
dimension by a zero-density state, then supplies the physical
interpretation of the branch and yields the causality
bound~\eqref{eq:causalitybound},
which resolves what would otherwise be a numerical curiosity:
$(N,d)=(3,10)$, the only subluminal case among the six plotted, is
precisely the saturating case $w_\infty=1$, every $N=3$ configuration
with $d>10$ being subluminal strictly. That $(2,6)$ and $(3,10)$ share
$k=-\tfrac12$ is, by contrast, an accident of the solvability
classification --- both belong to the quadratic family $d=4N-2$ ---
and not a causal one, since $(2,6)$ has $w_\infty=3$ and is
superluminal in the far field. Finally, the compactified phase portrait
identifies the basin of $Z^*$ as the interval $Z^Q_\infty<Z<1$ between
the pole and the repeller. Data below the pole and data above $Z=1$
both escape, the latter being the $a<0$ orientation excluded on
physical grounds in Sect.~\ref{ssec:central}; within the physical arc,
however, no fine-tuning whatever is required, and every initial
condition in $(Z^*,1)$ produces a regular unbounded model.

Physically, the two branches divide the conformally flat class along a
line that is easy to state. One describes bounded stars and nothing
else; the other describes extended halo-like distributions and can
never describe a star. Theorem~\ref{thm:unbounded} makes the
separation exact rather than a matter of parameter choice, and it is
the repeller at $Z=1$, with its universal eigenvalue $\lambda_1=1$,
that enforces it: a single trajectory cannot both leave a regular
centre and return to zero pressure at finite radius.

The attractor result gives the second branch a definite physical
character. Because $\lambda_*<0$ for every admissible pair, the
approach to the isothermal profile is a genuine relaxation: a
conformally flat configuration retains no memory of its central data in
the far field beyond an overall scale, and the closed form
$\lambda_*=-(d-2)/(d-2N)$ says how fast that memory decays. This is the
sense in which the present work goes beyond~\cite{Dadhich2016a}: the
pure Lovelock isothermal sphere was already known to be a universal
solution, and is shown here to be a universal \emph{endpoint}. The
dependence is instructive. At fixed $N$ the rate weakens as $d$ grows,
so higher-dimensional configurations approach the halo profile more
slowly; at fixed $d$ it strengthens with $N$, the higher-curvature term
driving the relaxation. In the critical dimension the ratio diverges
and the approach is fastest of all, which is the dynamical counterpart
of the anomalous density tail derived in Sect.~\ref{ssec:isothermal}.
That the limiting equation of state $w_\infty$ is fixed by $(d,N)$
alone, with no freedom left over, is the sense in which conformal
flatness selects a preferred asymptotic state.

A secondary result of independent interest is
equation~\eqref{eq:pressure_explicit}: on the generic branch the
pressure is determined solely by the spatial potential, the temporal
potential playing no dynamical role. This decoupling parallels a known
feature of the Finch--Skea solution~\cite{Finch1989,Hansraj2015b} and
suggests a structural principle specific to conformally flat classes.

Several directions remain. The exact radical models of the solvable
families, together with their physical analysis, are developed
in~\cite{HansrajCompanion}. We note that by
Theorem~\ref{thm:unbounded} the generic branch admits no finite
pressure-free surface, so smooth Darmois matching to a vacuum exterior
is unavailable there; any junction construction on this branch requires
either a thin shell or a non-vacuum exterior. Further extensions
include anisotropic fluids, charged solutions, and dynamical
(non-static) conformal flatness in the cosmological setting.

\appendix
\section{Algebraic degree and the Galois obstruction}
\label{app:galois}

This appendix supplies the proofs underlying
Sect.~\ref{ssec:algebraic}: that the algebraic degree of the spatial
potential is exactly $u+v$, and that four representative higher-degree
configurations have full symmetric Galois group and therefore admit no
radical inversion.

We first record the derivation of~\eqref{eq:algebraic}. Raising the
real integral~\eqref{eq:implicit_abs} to the power $v$, with $2k=-u/v$
in lowest terms and $u,v$ positive integers (recall $k<0$ in every
physical case), clears the fractional exponents and gives
\begin{equation}
  |Z-1|^{u+v}=c_1^{\,v}\,x^{u+v}\,|L(Z)|^{u}.
\label{eq:algebraic_abs}
\end{equation}
On each connected real sector of the phase line the signs of $Z-1$ and
of $L(Z)$ are constant, so the absolute values resolve into a single
non-zero constant and~\eqref{eq:algebraic_abs}
becomes~\eqref{eq:algebraic}, the constant $c$ depending on the sector.
Because this derivation starts from~\eqref{eq:implicit_abs} rather than
from the physical-arc form, it applies on \emph{every} real generic
sector and not only on $Z^*<Z<1$; the sector enters solely through the
sign of $c$, to which the degree and irreducibility arguments below are
insensitive. The relation is polynomial in $Z$ of degree $u+v$; that
this is also the degree of the \emph{minimal} polynomial requires
irreducibility, which the following lemma supplies.

\begin{lemma}[Irreducibility and degree]
\label{lem:irreducible}
For every admissible $(d,N)$ the polynomial~\eqref{eq:algebraic} is
irreducible over $K(x)$, where $K=\mathbb{Q}(c)$ is the field generated
over $\mathbb{Q}$ by the constant of the family; the scaling freedom of
Lemma~\ref{lem:scaling} may be used to normalise $|c|=1$, in which case
$K=\mathbb{Q}$. The spatial potential is consequently an algebraic
function of $x$ of degree exactly $u+v$, and the generic conformally
flat pure Lovelock metric is never transcendental in $x$.
\end{lemma}

\begin{proof}
Write $n=u+v$ and $L(Z)=(2k-1)^2Z-4k^2$, and regard~\eqref{eq:algebraic}
as a polynomial in $x$ over $K(Z)$, namely $x^n-a$ with
$a=(Z-1)^n/[c\,L(Z)^u]$. By Capelli's theorem $x^n-a$ is irreducible
over $K(Z)$ provided $a$ is not a $p$th power in $K(Z)$ for any prime
$p\mid n$, and, when $4\mid n$, provided $a\notin-4\,K(Z)^4$. The valuation of $a$ at the root of $L$
is $-u$; were $a$ a $p$th power, $p$ would divide $u$, and since
$p\mid n=u+v$ it would divide $\gcd(u,v)=1$, a contradiction. The same
valuation argument excludes the exceptional quartic case. Hence
$x^n-a$ is irreducible over $K(Z)$. Moreover $Z=1$ is not a
root of $L$, since $L(1)=1-4k\neq0$ for $k<0$, so $(Z-1)^n$ and
$L(Z)^u$ are coprime and~\eqref{eq:algebraic} is primitive in
$K[Z][x]$. Gauss's lemma then transfers irreducibility to $K[Z,x]$ and
hence to $K(x)[Z]$. \qed
\end{proof}

A closed form by radicals is \emph{guaranteed} whenever this degree does
not exceed four. Enumerating the coprime pairs with $u+v\le4$ gives the
complete list of such families,
$k\in\{-\tfrac14,-\tfrac12,-1,-\tfrac32,-\tfrac16\}$, corresponding
respectively to $d=3N-1$, $d=4N-2$, $d=6N-4$, $d=8N-6$ and $3d=8N-2$,
the last requiring $N\equiv1\pmod 3$ for $d$ to be an integer and so
first arising physically at $(N,d)=(4,10)$. We emphasise that degree at
most four is sufficient but not necessary: a polynomial of degree five
or more may still have a solvable Galois group, so the list above is
where radicals are assured, not where they are possible. These include
the Gauss--Bonnet configurations
$d=5$ (cubic in $Z$) and $d=6$ (quadratic) and the cubic Lovelock
configuration $d=10$ (quadratic) among the cases studied in
Sect.~\ref{sec:physical}. The resulting exact models and their physical
analysis are the subject of a companion
paper~\cite{HansrajCompanion}; here we retain the implicit and
parametric descriptions, which treat all $(d,N)$ uniformly.
Table~\ref{tab:cases} records the algebraic data for representative
cases.

\begin{proposition}[Insolubility by radicals]
\label{prop:galois}
For $(N,d)=(2,7)$, $(3,7)$, $(3,9)$ and $(4,9)$ the generic-branch
potential $Z(x)$ admits no closed-form expression by radicals. The
Galois group of~\eqref{eq:algebraic} is $S_5$ in the first two cases
and $S_7$ in the last two.
\end{proposition}

\begin{proof}
In each case we first clear the rational prefactor of $L(Z)$
in~\eqref{eq:algebraic}, absorbing it together with $c$ into a single
constant $\tilde{c}$, and then use the scaling freedom of
Lemma~\ref{lem:scaling} to set $\tilde{c}=-1$ on the physical sector,
the sign being the one dictated by $(1-Z)^{u+v}=|c|\,x^{u+v}L(Z)^u$
with $u+v$ odd. The coefficients then lie in $\mathbb{Q}(x)$. For
$(N,d)=(2,7)$ one has $k=-\tfrac34$ and~\eqref{eq:algebraic} becomes
the quintic $(Z-1)^5=-x^5(25Z-9)^3$; for $(3,7)$, $k=-\tfrac18$, and
the relation is $(Z-1)^5=-x^5(25Z-1)$. Now specialise the radial
variable $x$ to a rational value at which the resulting quintic remains
separable and irreducible. The Galois group of the specialised
polynomial is then a subgroup of the Galois group of the generic
polynomial over $\mathbb{Q}(x)$, so exhibiting a single specialisation
with group $S_5$ forces the generic group to be $S_5$, that being the
full symmetric group on five letters. We specialise $x$ itself rather
than any composite of it, so that the conclusion concerns the function
field actually appearing in~\eqref{eq:algebraic}.

Taking $x=1$ in the first case and $x=2$ in the second, the specialised
quintics are
\begin{align}
  f_1(Z)&=(Z-1)^5+(25Z-9)^3
        \nonumber\\
        &=Z^5-5Z^4+15635Z^3-16885Z^2+6080Z-730,
\label{eq:f1}\\
  f_2(Z)&=(Z-1)^5+32(25Z-1)
        \nonumber\\
        &=Z^5-5Z^4+10Z^3-10Z^2+805Z-33 .
\label{eq:f2}
\end{align}
Both are irreducible over $\mathbb{Q}$, with discriminants
$2^{42}5^{5}\cdot17\cdot4967$ and $2^{32}5^{5}\cdot13\cdot487$
respectively. Reducing modulo primes dividing neither the discriminant
nor the leading coefficient, so that Dedekind's theorem applies, the
factorisation degree patterns are
\begin{align}
  f_1 &: \;\text{irreducible mod }29,\quad 1{+}1{+}1{+}2\;\text{mod }149,
  \nonumber\\
  f_2 &: \;\text{irreducible mod }7,\;\;\;\;\; 1{+}1{+}1{+}2\;\text{mod }3,
\label{eq:certificates}
\end{align}
all four reductions being separable. Each first entry supplies a
$5$-cycle and each second a transposition. A transitive subgroup of
$S_5$ containing a transposition and a $5$-cycle is all of $S_5$, which
is not solvable; hence no expression of $Z$ in radicals exists in
either case.

Finally, the conclusion is independent of the sign convention, and so
covers the orientation $Z>1$ as well as the physical arc. Since
$u+v=5$ is odd, the map $x\mapsto-x$ reverses the sign of $x^{5}$ and
therefore carries~\eqref{eq:algebraic} with one choice of sign into the
same relation with the other. As $x\mapsto-x$ is an automorphism of
$\mathbb{Q}(x)$ fixing $\mathbb{Q}$, the two polynomials are conjugate
over $\mathbb{Q}(x)$ and have the same generic Galois group.

For $(N,d)=(3,9)$ and $(4,9)$ the degree is $u+v=7$, and
\eqref{eq:algebraic} at $x=1$ gives, in the physical sign,
\begin{align}
  g_1(Z)&=(Z-1)^7+(49Z-9)^3,
\label{eq:g1}\\
  g_2(Z)&=(Z-1)^7+(49Z-1),
\label{eq:g2}
\end{align}
Irreducibility over $\mathbb{Q}$ is immediate from the reductions
$g_1$ mod $11$ and $g_2$ mod $5$, both of good reduction and both
irreducible of degree seven; hence both Galois groups are transitive of
degree seven. The transitive subgroups of $S_7$ have
orders $7$, $14$, $21$, $42$, $168$, $2520$ and $5040$. Reduction
modulo $13$ for $g_1$ and modulo $29$ for $g_2$, both primes of good
reduction, gives in each case the factorisation type $2+5$: an element
of order ten and odd parity. Only $2520$ and $5040$ are divisible by
ten, and the group of order $2520$ is $A_7$, which contains no odd
permutation; hence both groups are $S_7$, which is not solvable.
Equivalently and independently, the factorisation type
$1{+}1{+}1{+}1{+}1{+}2$ occurs at $p=2011$ for $g_1$ and $p=2113$ for
$g_2$, exhibiting a transposition; a transitive group of prime degree
is primitive, and by Jordan's theorem a primitive group containing a
transposition is the full symmetric group. \qed
\end{proof}

\begin{remark}
The primes in~\eqref{eq:certificates} must be chosen with care, since a
degree multiset alone does not certify a cycle type. For the
opposite-sign quintic $(Z-1)^5-(25Z-9)^3$, for instance, the reduction
modulo $59$ is $(Z+19)(Z-23)^2(Z^2+22Z+15)$, whose degrees are
nominally $1{+}1{+}1{+}2$; but $59$ divides that discriminant, the
reduction is not separable, and Dedekind's theorem does not apply.
Every prime used in~\eqref{eq:certificates} is one of good reduction.
\end{remark}

\begin{acknowledgments}
The author thanks the University of KwaZulu-Natal for its continued
support.
\end{acknowledgments}

\section*{Declarations}

\noindent\textbf{Funding.} The author declares that no funds, grants or
other support were received during the preparation of this
manuscript.

\medskip\noindent\textbf{Competing interests.} The author has no
relevant financial or non-financial interests to disclose.

\medskip\noindent\textbf{Data availability statement.} This manuscript
has no associated data.  All results are analytic or
are reproduced by direct numerical integration of
equation~(\ref{eq:master}) using the initial data and parameter values stated in Sect.~\ref{ssec:numerical}; no external data sets were generated or analysed.

\medskip\noindent\textbf{Code availability.} The \textsc{Mathematica}
notebook used to generate Figs.~\ref{fig:phase}--\ref{fig:v2} is
available from the author on reasonable request.


\end{document}